\documentclass[sn-mathphys,Numbered,dvipsnames]{sn-jnl}% Math and Physical Sciences Numbered Reference Style 
\usepackage{graphicx}%
\usepackage{multirow}%
\usepackage{amsmath,amssymb,amsfonts}%
\usepackage{amsthm}%
\usepackage{mathrsfs}%
\usepackage[title]{appendix}%
\usepackage{xcolor}%
\usepackage{textcomp}%
\usepackage{manyfoot}%
\usepackage{booktabs}%
\usepackage{algorithm}%
\usepackage{algorithmicx}%
\usepackage{algpseudocode}%
\usepackage{listings}%
\usepackage[utf8]{inputenc}
\usepackage[T1]{fontenc}

\usepackage{amsmath}
\usepackage{amssymb}
\usepackage{caption}
\usepackage{subcaption}

\usepackage{xspace}

\usepackage{colortbl}

\usepackage{colortbl}

\usepackage{booktabs}

\usepackage{tikz}
\usetikzlibrary{math, automata, positioning}

\usepackage{pgfplots}
\pgfplotsset{compat=1.18}
\usepackage{bbding}

\newcommand{\myquot}[1]{``#1''}

\newcommand{\nats}{\mathbb{N}}
\newcommand{\reals}{\mathbb{R}}
\newcommand{\nnreals}{\reals_{\geq 0}}
\newcommand{\rats}{\mathbb{Q}}
\newcommand{\nnrats}{\rats_{\geq 0}}

\newcommand{\auta}{\mathcal{A}}
\newcommand{\aut}{\mathcal{B}}
\newcommand{\lang}[1]{L(#1)}
\newcommand{\infset}[1]{\mathrm{Inf}(#1)}

\newcommand{\transition}[1]{\xrightarrow{(\sigma_{#1},\tau_{#1})}}
\newcommand{\transitiondelta}[1]{\xrightarrow{(\sigma_{#1},\tau_{#1}+\delta)}}
\newcommand{\transitiondeltam}[1]{\xrightarrow{(\sigma_{#1},\tau_{#1}-\delta)}}

\newcommand{\labeledtransition}[1]{\xrightarrow{#1}}

\newcommand{\nonempty}[1]{S_{#1}^{\textit{ne}}}
\newcommand{\terminal}[2]{\mathcal{T}_{#1}(#2)}

\newcommand{\tba}{TBA\xspace}
\newcommand{\mitl}{MITL\xspace}

\newcommand{\true}{\texttt{true}}
\newcommand{\false}{\texttt{false}}

\makeatletter
\newcommand{\shortsim}{%
  \settowidth{\@tempdima}{n}% Width of n
  \resizebox{\@tempdima}{\height}{$\sim$}%
}
\makeatother

\newcommand{\TSigma}{T\Sigma}

\newcommand{\pow}[1]{2^{#1}}
\newcommand{\size}[1]{|#1|}

\newcommand{\bools}{\mathbb{B}}

\newcommand{\assumpverdicts}{\bools_4}
\newcommand{\unknown}{\textbf{\textit{?}}}
\newcommand{\outofmodel}{\times}

\newcommand{\verdle}{\preccurlyeq}
\newcommand{\verdge}{\succcurlyeq}

\newcommand{\assumpevalfuncsymbol}{\mathcal{V}\xspace}
\newcommand{\assumpevalfunc}[2]{\assumpevalfuncsymbol(#1)(#2)}

\newcommand{\conswords}[2]{\mathcal{C}_{#1}(#2)}
\newcommand{\obs}{o}

\newcommand{\post}{\mathrm{Post}}
\newcommand{\succc}{\mathrm{Succ}}

\newcommand{\multbound}{e}

\usepackage[color=Cerulean!30]{todonotes}

\newtheorem{definition}{Definition}
\newtheorem{example}{Example}
\newtheorem{theorem}{Theorem}
\newtheorem{lemma}{Lemma}
\newtheorem{remark}{Remark}
\newtheorem{proposition}{Proposition}

\begin{document}

\title{Exploiting Assumptions for Effective Monitoring of Real-Time Properties under Partial Observability}

\author[1]{\fnm{Alessandro} \sur{Cimatti}}\email{cimatti@fbk.eu}
\author*[2]{\fnm{Thomas M.} \sur{Grosen}}\email{tmgr@cs.aau.dk}
\author[2]{\fnm{Kim G.} \sur{Larsen}}\email{kgl@cs.aau.dk}
\author[1]{\fnm{Stefano} \sur{Tonetta}}\email{tonettas@fbk.eu}
\author[2]{\fnm{Martin} \sur{Zimmermann}}\email{mzi@cs.aau.dk}

\affil[1]{\orgname{Fondazione Bruno Kessler}, \orgaddress{\city{Trento}, \country{Italy}}}

\affil[2]{\orgname{Aalborg University}, \orgaddress{\city{Aalborg}, \country{Denmark}}}

%% \abstract{Runtime verification of temporal properties over timed sequences of observations is crucial in various applications within cyber-physical systems ranging from autonomous vehicles over smart grids to medical devices. In this paper, we are addressing the challenge of effectively predicting  the failure or success of properties in a continuous real-time setting. Our approach allows predictions to exploit assumptions on the system being monitored and supports predictions of  non-observable system behaviour (e.g., internal faults).  More concretely, in our approach properties are expressed  in Metric Interval Temporal Logic (MITL), assumptions on the monitored system are specified in terms of Timed Automata, and observations are to be provided in terms of sequences of timed constraints. We present an assumption-based runtime verification algorithm and its implementation on top of the real-time verification tool UPPAAL. We show experimentally that assumptions can be effective in anticipating the satisfaction/violation of timed properties and in handling monitoring properties that predicate over unobservable~events. 

\abstract{
Runtime verification of temporal properties is essential for ensuring
the correctness and reliability of real-time systems, particularly in
cyber-physical systems. A significant challenge in this domain is the
effective prediction of property failure or success, especially when
dealing with partially observable systems. This paper addresses these
challenges by developing an Assumption-Based Runtime Verification
(ABRV) approach for a continuous real-time setting. Our method
exploits assumptions about the system's behavior, specified as Timed
Automata, to enable monitors to predict future outcomes and handle
unobservable system parts, such as internal faults. Properties to be
monitored are specified using Metric Interval Temporal Logic
(MITL). The approach also includes formalizing observations with data
and time uncertainty using sequences of timed constraints.  We present
a zone-based online algorithm for computing the monitoring verdict,
implemented on top of the UPPAAL tool. Experimental evaluation on
proof-of-concept cases demonstrates the approach's feasibility and
effectiveness, illustrating how assumptions facilitate earlier
verdicts, enable monitoring of properties dependent on unobservable
events, and provide insights into scalability.
}


\maketitle

%%%%%%%%%%%%%%%%%%%%%%%%%%%%%%%%%%%%%%%%%%%%%%%%%%%%%%%%%%%%%%%%%%%%%%%%%%%
%%%%%%%%%%%%%%%%%%%%%%%%%%%%%%%%%%%%%%%%%%%%%%%%%%%%%%%%%%%%%%%%%%%%%%%%%%%
%%%%%%%%%%%%%%%%%%%%%%%%%%%%%%%%%%%%%%%%%%%%%%%%%%%%%%%%%%%%%%%%%%%%%%%%%%%
\section{Introduction}
\label{sec:intro}

The problem of monitoring timed properties has gained significant
attention due to its crucial role in ensuring the correctness and
reliability of real-time systems. The runtime verification of temporal
properties over timed sequences of observations is crucial in various
applications ranging from cyber-physical systems including autonomous
vehicles and beyond. While different solutions for
runtime verification of timed temporal properties have been presented~\cite{bauer2006monitoring,basin2012algorithms,baldor2013monitoring,ho2014online,GrosenKLZ22},
some challenges remain to be addressed, in particular extending these solutions
with prognosis and diagnosis capabilities. More specifically, we are here
interested in effectively predicting in advance the failure of
properties and in handling partially observable systems.

In the discrete-time setting, these challenges have been addressed
with As\-sump\-tion-Based Runtime Verification
(ABRV)~\cite{Tian:2019a,Tian:2019b,CimattiTT21,CimattiTT22}.  ABRV uses
assumptions about the behavior of the system to predict the future
behavior of the system and to relate observable and non-observable
variables. These assumptions can be derived, for example, from models
produced during the system design, or from the data collected from the
system in operation. Exploiting assumptions, the monitor can
anticipate the detection of property failures. Moreover, the
specification is no more limited to the interface of black box systems
as in traditional runtime verification, but can be extended to
constrain also the internal non-observable parts (such as, for example,
internal faults).

Thus, in this setting, partial observability means that the monitor
may not receive information about the occurrence of some events (e.g.,
internal faults or signals not exposed at the system interface). The
knowledge on the system execution is therefore underspecified: an
observation constrains what must have happened, but leaves
unobservable events unconstrained. Assumptions are then used to
reconstruct and constrain the possible behaviors of the system
compatible with the observed events. The assumption alphabet includes
observable and unobservable events and, when combined with the
observation, generates all executions consistent with what was
seen. This allows the monitor to reason about properties that refer to
unobservable events and to issue definitive verdicts even when such
events were never directly observed.

In ABRV, the output of the monitor has four possible values:
\begin{itemize}
\item $\top$ (Satisfied): given the sequence of observations, the system satisfies
the specified temporal properties under the given assumption.
\item $\bot$ (Violated): this value indicates that the observed behavior of the
system violates the specified temporal property, under the given
assumption.
\item $\outofmodel$ (Out-of-model): the observed behavior
violates the assumptions, i.e., there is no run of the assumption compatible
with the observations.
\item $\unknown$ (Unknown): given the current observations and assumption, it
is not possible to determine definitively whether the property is
satisfied or violated. 
\end{itemize}

In this work, we enhance ABRV for real-time systems, where:
\begin{itemize}
  \item
    \textit{Properties} are expressed in Metric Interval Temporal
    Logic (MITL).
  \item
    \textit{Assumptions} are specified using Timed Automata.
  \item
    \textit{Observations} incorporate data and time uncertainty, captured via
    sequences of timed constraints.
\end{itemize}

We formally define the semantics of monitoring under assumptions in
this setting and prove that our monitoring function satisfies
desirable properties: \textit{impartiality}, once a definitive verdict
is issued, it remains valid for all future observation extensions;
\textit{anticipation}, a definitive verdict is issued as early as
possible when sufficient evidence exists; we introduce a partial order
on the four verdicts $\{\top,\bot, \outofmodel,\unknown\}$, formalizing
how assumptions and observations refine the verdict.

Like in the discrete-time case~\cite{CimattiTT22}, the assumption allows the monitor to
give a $\top$ or $\bot$ verdict even if the property contains future
operators and non-observable events. For example, suppose we monitor
the MITL property~$\varphi = F_{[0,10]}a \land G_{[0,20]}\neg b$ (expressing that there is an ``$a$'' in the first ten units of time, but no ``$b$'' in the first 20 units of time) and
we assume that the system satisfies the assumption~$\psi=G_{[0,1]}\neg b \wedge G(a \rightarrow
G_{[0,10]} \neg b)$ (expressing that there is no ``$b$'' in the first unit of time and no ``$a$'' is followed by a ``$b$'' within ten units of time). Then, the monitor can output a $\top$ verdict
even before time $20$, for instance at time $10$ when ``$b$'' is false in
the interval $[0,10]$ and ``$a$'' is true at time $10$. Further, it
can even give the verdict $\top$ if ``$b$'' is not observable, e.g.,
when ``$a$'' is true at time $0$ and $10$.

One of our main contributions is a rich definition of
observations that take into account both data and time uncertainty.
As in the discrete case~\cite{CimattiTT22}, the observations are represented by formulas that can
capture the uncertainty on data. For example, $\neg a$ means that $a$
is not seen but ``$b$'' can be true or false. The approach is further
extended to have uncertainty on time, taking into account potential
errors in the timestamps with which the monitor receives data from the
system. This is represented in the observations with time intervals
that are associated to observation formulas. Thus, for example, we can
say that ``$a$'' is seen in the interval~$[6,7]$ but we do not know exactly
when.
Also, we support multiplicities to express bounds on the number of occurrences of events.
Finally, we concatenate triples of formulas, time
intervals, and multiplicities to form complex sequences of observations. For example, the
sequence 
\[(a,[0,0],=1)\,(\neg a,[0,7],\ge 0)\,(a,[6,7],=1)\,(\neg
a,[6,16],\ge 0)\,(a,[15,16],=1)\]
says that we see three $a$'s, one at time
$0$, another in the interval $[6,7]$, and a final one in the interval
$[15,16]$ and that we do not know anything about ``$b$'' in between the three ``$a$'' (intuitively, an observation with an $=1$ ($\ge 0$) indicates exactly one occurrence (zero or more occurrences)). If the system
satisfies the assumption~$\psi$ from above, we can conclude at time $16$ that the property~$\varphi$ is true despite the uncertainty about time and $b$.

We propose a \textit{zone-based} \textit{online} algorithm that at any
time provides a monitoring verdict saying if the property is
satisfied or violated given the assumption and a sequence of
observations.
We implemented the algorithm on top of UPPAAL and show the
feasibility of the approach. Especially, we demonstrate how the
assumptions can be effective in anticipating the
satisfaction/violation of timed properties and in handling
properties that predicate over unobservable events.
We also report on the influence of unobservable events on the response time, the time it takes to compute a verdict when given a new observation.

This paper significantly extends the work presented in
\cite{CimattiGLTZ24} by providing more formal theoretical
underpinnings, e.g., the results in Subsection~\ref{sec:abstract} showing that our definition is anticipatory, impartial, and monotone in all parameters, as well as a new consistency check for observations (Lemma~\ref{lemma_obsemptinesstest}).
Furthermore, it enhances
the definition and handling of observations by incorporating a richer
set of multiplicities, i.e., from \myquot{exactly one} and \myquot{any number} to \myquot{exactly $k$}, \myquot{at least $k$}, and \myquot{at most $k$} for every $k \in \nats$. 
Among other things, this now allows to express upper bounds. 
Similarly, the original monitoring algorithm has been extended to support such
multiplicities.
Finally, we present an
additional case study in Subsection~\ref{sec:jobshop} along with more detailed experimental results
for the previous case studies.

The remainder of the paper is organized as
follows. Section~\ref{sec:prelim} introduces the necessary
preliminaries, including timed words, Timed Automata, and
MITL. Section \ref{sec:monitoringunderassumptions} defines the
monitoring problem under assumptions and presents the formal
semantics, including proofs of impartiality, anticipation, and the
partial order on verdicts. Section \ref{sec:impl} develops a zone-based algorithm
for monitoring under assumptions, handling uncertainty in both data
and time.  Section \ref{sec:eval} describes our implementation built on UPPAAL and
presents experimental results that demonstrate feasibility and the
benefits of assumption-based monitoring.  Section \ref{sec:relatedwork} discusses related
work, and Section \ref{sec:conc} concludes with future research directions.

%%%%%%%%%%%%%%%%%%%%%%%%%%%%%%%%%%%%%%%%%%%%%%%%%%%%%%%%%%%%%%%%%%%%%%%%%%%
%%%%%%%%%%%%%%%%%%%%%%%%%%%%%%%%%%%%%%%%%%%%%%%%%%%%%%%%%%%%%%%%%%%%%%%%%%%
%%%%%%%%%%%%%%%%%%%%%%%%%%%%%%%%%%%%%%%%%%%%%%%%%%%%%%%%%%%%%%%%%%%%%%%%%%%
\section{Preliminaries}
\label{sec:prelim}

The set of natural numbers (excluding zero) is $\nats$, we define~$\nats_0 = \nats \cup \{0\}$, the set of non-negative rational numbers is $\nnrats$, and the set of non-negative real numbers  is $\nnreals$. 
The powerset of a set $S$ is denoted by $\pow{S}$.

\subsection{Timed Words}
A timed word over a finite alphabet~$\Sigma$ is a pair $\rho = (\sigma, \tau)$ where $\sigma$ is a word over $\Sigma$ and $\tau$ is a sequence of non-decreasing non-negative real numbers of the same length as $\sigma$.
Timed words may be finite or infinite. In the latter case, we require $\limsup \tau = \infty$, i.e., time diverges. 
The set of finite timed words is denoted by $\TSigma^*$ and the set of infinite timed words by $\TSigma^\omega$.
We also represent a timed word as a sequence of pairs~$(\sigma_1,\tau_1) (\sigma_2,\tau_2) \ldots$.  If $\rho=(\sigma_1,\tau_1) (\sigma_2,\tau_2) \cdots (\sigma_n,\tau_n)$ is a finite timed word, we denote by $\tau(\rho)$ the total time duration of $\rho$, i.e., $\tau_n$ (with the convention~$\tau(\varepsilon) = 0$).
We lift this definition to languages~$L\subseteq \TSigma^*$ by defining~$\tau(L) = \sup_{\rho \in L}\tau(\rho)$, which can be infinite.

If $\rho_1=(\sigma^1_1,\tau^1_1)(\sigma_2^1,\tau_2^1)\ldots(\sigma^1_n,\tau^1_n)$ is a finite timed word, $\rho_2=(\sigma^2_1,\tau^2_1)(\sigma^2_2,\tau^2_2)\ldots$ a finite or infinite timed word, and $t \in \nnreals$ then the concatenation $\rho_1 \cdot_t \rho_2$ is defined if and only if $t \ge \tau(\rho_1)$. Then, we define~$\rho_1 \cdot_t \rho_2 = (\sigma_1,\tau_1)(\sigma_2,\tau_2)\cdots$ such that 
\[\sigma_i = \begin{cases}
    \sigma^1_i & \text{if } i\le n\\
    \sigma^2_{i-n} & \text{else}
\end{cases}\quad\text{and}\quad
\tau_i = \begin{cases}
    \tau^1_i & \text{if } i\le n\\
    \tau^2_{i-n}+t & \text{else}.
\end{cases}\]
We lift this definition to sets~$L_1 \subseteq \TSigma^*$ and $L_2 \subseteq \TSigma^* \cup \TSigma^\omega$ via
\[
L_1 \cdot_t L_2 = \{\rho_1 \cdot_t \rho_2 \mid \rho_1 \in L_1 \text{ and } \rho_2 \in L_2\},
\]
provided we have $t \ge \tau(L_1)$.

\begin{remark}
\label{remark_concatprops}
The following two properties, which follow directly from the definition of concatenation, will be useful later on.\footnote{Note that we could formulate the properties more generally by concatenating with languages other than $\TSigma^\omega$, but we refrain from doing so, as this is the result we need later on.}
Let $L,L' \subseteq \TSigma^*$.
\begin{enumerate}
    \item If $L \subseteq L'$ and $\tau(L), \tau(L') \le t$, then $L \cdot_t \TSigma^\omega \subseteq L' \cdot_t \TSigma^\omega$.
    \item If $\tau(L) \le t \le t'$, then $L \cdot_{t'} \TSigma^\omega \subseteq L \cdot_{t} \TSigma^\omega$.
\end{enumerate}
\end{remark}

\subsection{Timed Automata}

A timed Büchi automaton~(\tba)~$\aut = (Q, Q_0, \Sigma, C, \Delta, \mathcal{F})$ consists of a finite alphabet~$\Sigma$, a finite set~$Q$ of locations, a set~$Q_0 \subseteq Q$ of initial locations, a finite set~$C$ of clocks, a finite set~$\Delta \subseteq Q \times Q \times \Sigma \times \pow{C} \times G(C)$ of transitions with $G(C)$ being the set of clock constraints over $C$, and a set~$\mathcal{F} \subseteq Q$ of accepting locations.
A transition~$(q,q',a,\lambda, g)$ is an edge from $q$ to $q'$ on input symbol~$a$, where $\lambda$ is the set of clocks to reset and $g$ is a clock constraint over $C$.
A clock constraint is a conjunction of atomic constraints of the form~$c \sim n$, where $c$ is a clock, $n \in \nats_0$, and $\sim\ \in \{<, \leq, =, \geq, >\}$.

A state of $\aut$ is a pair $(q,v)$ where $q$ is a location in $Q$ and $v \colon C \rightarrow \nnreals$ is a valuation mapping clocks to their values.
For any $d \in \nnreals$, $v+d$ is the valuation~$x \mapsto v(x) + d$.
A run of $\aut$ from a state $(q_0, v_0)$  over a timed word $(\sigma_1, \tau_1)  (\sigma_2, \tau_2)\cdots$ is a sequence
\[(q_0, v_0) \transition{1} (q_1,v_1) \transition{2} (q_2,v_2) \transition{3} \cdots \] 
where for all $i \geq 1$ there is a transition $(q_{i-1},q_{i},\sigma_{i},\lambda_i,g_i)$ such that $v_{i}(c) = 0$ for all $c$ in $\lambda_i$ and $v_{i}(c) = v_{i-1}(c) + (\tau_i - \tau_{i-1})$ otherwise, and $g_i$ is satisfied by the valuation $v_{i-1}+(\tau_{i} - \tau_{i-1})$.
Here, we use $\tau_0 = 0$.
Given a run $r$, we denote the set of locations visited infinitely many times by $r$ as $\infset{r}$.
A run $r$ of $\aut$ is accepting if $\infset{r} \cap \mathcal{F} \neq \emptyset$.
The language of $\aut$ from a starting state~$(q,v)$, denoted $\lang{\aut,(q,v)}$, is the set of all timed words with an accepting run in $\aut$ starting from $(q,v)$. 
We define the language of $\aut$, written $\lang{\aut}$, to be $\bigcup_q \lang{\aut,(q,v_0)}$, where $q$ ranges over all locations in $Q_0$ and where $v_0(c) = 0$ for all $c \in C$.

\begin{proposition}[\cite{alur1994tba}]
\label{prop:intersection}
For all \tba $\aut$, $\aut'$ there is a \tba~$\aut \otimes \aut'$ with $\lang{\aut \otimes \aut'}=\lang{\aut}\cap\lang{\aut'}$. The set of states of $\aut \otimes \aut'$ is $Q \times Q' \times  \{0,1\}$, where $Q$ and $Q'$ are the sets of states of $\aut$ and $\aut'$, respectively.
\end{proposition}

\subsection{Logic}
We use Metric Temporal Interval Logic (\mitl) to formally express properties to be monitored in our examples. These can be translated into equivalent \tba which we use in our monitoring algorithm.

The syntax of \mitl formulas over a finite alphabet~$\Sigma$ is defined as
\[
\varphi ::= p \mid  \neg \varphi \mid  \varphi \vee \varphi \mid  X_I \varphi \mid  \varphi\ U_{I} \varphi ,
\]
where $p \in \Sigma$ and $I$ ranges over non-singular intervals over $\nnreals$ with endpoints in ${\nats_0 \cup \{\infty\}}$.
Note that we often write $\sim\,n$ for $I=\{d\in\reals\mid d\sim n\}$ where $\sim\ \in \{<,\leq,\geq,>\}$, and $n \in \nats$.
  We also define the standard syntactic sugar
$\true = p \vee \neg p$,
$\false = \neg \true$,
$\varphi \wedge \psi = \neg (\neg \varphi \vee \neg \psi)$,
$\varphi \rightarrow \psi = \neg \varphi \vee \psi$,
$F_{I} \varphi = \true\ U_{I} \varphi$, 
and $G_{I} \varphi = \neg F_{I} \neg \varphi$.
Finally, we often omit the interval~$[0,\infty]$ for readability.

The semantics of \mitl is defined over infinite timed words.
Given such a timed word $\rho = (\sigma_1,\tau_1)(\sigma_2,\tau_2)\cdots \in \TSigma^\omega$, a position $i \geq 1$, and an \mitl formula~$\varphi$, we inductively define the satisfaction relation $\rho, i \models \varphi$ as  follows:
\begin{itemize}
    \item $\rho,i \models p $  if and only if $p = \sigma_i$.
    \item $\rho,i \models \neg \varphi$ if and only if $\rho,i \not\models \varphi$.
    \item $\rho,i \models \varphi \vee \psi$  if $\rho,i \models \varphi$ or $\rho,i \models \psi$.
    \item $\rho,i \models X_I \varphi$ if and only if $\rho,(i+1) \models \varphi$ and $\tau_{i+1} - \tau_i \in I$.
    \item $\rho,i \models \varphi\ U_{I} \psi$ if and only if  there exists $k \geq i$ such that $\rho,k \models \psi$, $\tau_k - \tau_{i} \in I$, and $\rho,j \models \varphi$ for all $i \leq j < k$. 
\end{itemize}
  We write $\rho \models \varphi$ whenever $\rho, 1 \models \varphi$, and say that $\rho$ satisfies $\varphi$.
The language $\lang{\varphi}$ of an \mitl formula~$\varphi$ is the set of all infinite timed words that satisfy~$\varphi$.

\begin{theorem}[\cite{alur1996mitl,brihaye2017mightyl}]
\label{thm:mtltba}
For each \mitl formula $\varphi$ there exists a \tba $\aut$ such that $\lang{\varphi} = \lang{\aut}$.
\end{theorem}  

\newcommand{\initsymb}{s}
\begin{example}\label{ex:mitl2automaton}
Figure~\ref{fig:aut-example} illustrates the above theorem by providing \tba for the formula~$\initsymb \land F_{[0,10]}a \land G_{[0,20]}\neg b$ and its negation. 
The formula requires the first symbol to be an ``$\initsymb$'', that an $a$ occurs within ten time units after the first symbol, and that no $b$ occurs within twenty time units after the first symbol, i.e., the ``$\initsymb$'' marks the initial timestamp of the processed word. 
\end{example}

\begin{figure}[bt]
    \centering
    \begin{tikzpicture} [node distance = 2.3cm,thick,>=stealth]
    \node (q0)     [state, initial text={}]          {$q_0$};
    \node (q1) [state, below = 1.25 of q0] {$q_1$};
    \node (i)   at (1,0) {};
    \node (q2)     [state, right = of q1]    {$q_2$};
    \node (phi)    [state, right = of q2]    {$\varphi$};
    \node (notphi) [state, left = of q1]    {$\neg\varphi$};
    
    \draw[->, thick] (i) edge (q0);
    \draw[->, thick] (q0) edge node[right, align=center] {$\initsymb$\\ $x :=  0$} (q1);
    \draw[->, thick] (q0) edge[bend right] node[above,xshift=-3pt] {$a,b$} (notphi);
     \draw[->, thick] (q1) edge node[above] {$a$} node[below]{$x \leq 10$} (q2);
     \draw[->, thick] (q2) edge node[above] {$a, b$} node[below]{$x > 20$} (phi);

     \draw[->, thick] (q1) edge[bend right=35] node[above]{$b$} (notphi);
     \draw[->, thick] (q1) edge[] node[above]{$a$} node[below]{$x > 10$} (notphi);

     \draw[->, thick] (q2) edge[bend left=45] node[above,near start]{$b$}  node[below,near start,yshift=-3pt]{$x\leq20$} (notphi);
    
    \draw[->, thick] (phi) edge [loop below] node[] {$a,b,\initsymb$} (phi);
    \draw[->, thick] (q2) edge [loop above] node[above,align=center] {$a$\\ $x \leq 20$} (q2);
    \draw[->, thick] (q2) edge [loop below] node[below,align=center] {$\initsymb$} (q2);
    \draw[->, thick] (q1) edge [loop below] node[below,align=center] {$\initsymb$} (q1);
    \draw[->, thick] (notphi) edge[loop below] node[]{$a,b,\initsymb$} (notphi);

\end{tikzpicture}
    \caption{A \tba for the languages of the formula~$\varphi = \initsymb \land F_{[0,10]}a \land G_{[0,20]}\neg b$ and its negation: If location~$\varphi$ ($\neg \varphi$) is accepting then it accepts $L(\varphi)$ ($L(\neg \varphi)$).} 
    \label{fig:aut-example}
\end{figure}

In the following, for the sake of readability, we use $\varphi$ also to denote properties, i.e., subsets of $\TSigma^\omega$, and use \mitl only to specify properties in examples.

%%%%%%%%%%%%%%%%%%%%%%%%%%%%%%%%%%%%%%%%%%%%%%%%%%%%%%%%%%%%%%%%%%%%%%%%%%%
%%%%%%%%%%%%%%%%%%%%%%%%%%%%%%%%%%%%%%%%%%%%%%%%%%%%%%%%%%%%%%%%%%%%%%%%%%%
%%%%%%%%%%%%%%%%%%%%%%%%%%%%%%%%%%%%%%%%%%%%%%%%%%%%%%%%%%%%%%%%%%%%%%%%%%%
\section{Monitoring under Assumptions}
\label{sec:monitoringunderassumptions}

Monitoring timed properties~\cite{bauer2006monitoring,GrosenKLZ22} requires to determine whether every extension of a finite observation (a finite timed word) satisfies a given property (yielding the verdict~$\top$), whether every extension violates the  property (yielding the verdict~$\bot$), or neither is true (yielding the verdict~$\unknown$).
In this section we introduce monitoring of real-time properties under assumptions and partial observability.

We begin in Subsection~\ref{sec:motivation} with an informal description of monitoring of real-time properties under assumptions and partial observability.
Then in Subsection~\ref{sec:abstract} we present an abstract definition capturing this description, and prove some basic properties about it. This abstract definition employs arbitrary timed languages for observations, assumptions, and specifications, and is therefore not algorithmically feasible.
Finally in Subsection~\ref{sec:concrete} we present an effective concretisation of the abstract definition, using timed automata and finite expressions to represent assumptions, specifications, and observations, respectively.

\subsection{Motivation}\label{sec:motivation}
Monitoring under assumptions involves two changes over the classical monitoring framework.

Firstly, the assumption itself: In its most general form, it is a set~$A \subseteq \TSigma^\omega$ of infinite timed words. Intuitively, $A$ contains the executions we assume to be possibly generated by the system we are monitoring. Hence, every execution that is not in $A$ does not need to be taken into account when determining a verdict, i.e., the assumption refines verdicts.
However, this also means that our assumption can be invalidated if we observe an execution prefix that is not consistent with our assumption. This requires a new verdict, denoted by $\outofmodel$. In this case, the assumption needs to be refined as it does not match our observation.

Secondly, we allow inexact observations: In the classical setting, we observe a finite timed word $(\sigma_1, \tau_1) \cdots (\sigma_n, \tau_n)$ and reason about its possible extensions. Hence, we implicitly presume that no other events occurred between time~$0$ and $\tau_n$ and that the time points are exact. 
In the following, we allow for some imperfect information about the observation. In the most general form, an observation is then a set~$O \subseteq \TSigma^*$ of finite timed words. Intuitively, $O$ contains those words that are consistent with our (imperfect) observation.

\begin{example}\label{ex:monitoring}
Consider the property ``$\initsymb  \land F_{[0,10]}a \land G_{[0,20]}\neg b$'' of
Example~\ref{ex:mitl2automaton}. Monitoring this property on a timed
word starting with ``$\initsymb$'' at time~$0$ provides a conclusive verdict in the following cases:
\begin{itemize}
\item
  The property is false at any time in the interval~$[0,20]$ if a ``$b$'' is observed.
\item
  The property is false after time~$10$ if ``$a$'' was not previously observed.
\item
  The property is true after time~$20$ if ``$b$'' was not previously  observed
  and ``$a$'' was observed in the interval
  $[0,10]$.
\end{itemize}
Note that a positive verdict can only given after time point~$20$, as any ``$b$'' before or at time point~$20$ implies that the property is violated.

Consider now the assumption ``$ \initsymb \land X G\neg\initsymb \land G_{[0,1]}\neg b \wedge G(a \rightarrow G_{[0,10]} \neg b)$'' which
corresponds to the \tba in Figure~\ref{fig:aut-example2}. Then,
if ``$\initsymb$'' was observed at time~$0$ and ``$a$'' was observed in the interval $[0,10]$, as
soon as we see another ``$a$'' at some time point $t \in [10,20]$ and have not observed a ``$b$'' before $t$, we can conclude
that the property is true already at time point~$t$, which may be strictly smaller than $20$.
In this case, the assumption leads to earlier verdicts.

On the other hand, if after the two ``$a$'', a ``$b$'' is observed in the interval~$[t,t+10]$, then the observation violates the assumption, even if the property may be satisfied (i.e., if the ``$b$'' is observed in the interval~$(20,t+10]$). 
\end{example}

\begin{figure}[bt]
    \centering
    \begin{tikzpicture} [node distance = 2.7cm,thick,>=stealth]
    \node (q0)  at (-3,0)   [state, initial text={}]          {$q_0$};
    \node (i)   at (-4,0) {};

    \node (qhalf) at (0,0) [state,accepting] {$q_1$};
    \node (q1)  at (4.5,0)   [state, accepting]    {$q_2$};
    \node (q2)  at (9,0)  [state]     {$q_3$};
    
    \draw[->, thick] (i) edge (q0);
     \draw[->, thick] (qhalf) edge [bend left=30] node[above] {$a$} node[below]{$y:=0$} (q1);

\draw[->, thick] (q0) edge[bend left] node[above,align=center] {$\initsymb$\\$x:=0$\\$y:=0$} (qhalf);

     \draw[->, thick] (q1) edge [loop right] node[right,align=center]{$a$\\$y:=0$} ();

     \draw[->, thick] (qhalf) edge [loop left] node[below,align=center]{$b$\\ $x>1$} (qhalf);

     \draw[->, thick] (q1) edge [bend left=30] node[above] {$b$} node[below]{$x > 1 \wedge y > 10$} (qhalf);

         \draw (qhalf.north) edge[->, thick] [bend left=30] node[above] {$b$} node[below]{$x \le 1$} (q2.north);
              \draw (q1.north east) edge[->, thick] [bend left=30] node[above] {$b$} node[below]{$x \le 1 $} (q2);
              \draw (q1.south east) edge[->, thick] [bend right=30] node[above] {$b$} node[below]{$ y \le 10$} (q2);

    \draw (q2) edge[loop right, thick, ->] node[right] {$a,b$} ();
\end{tikzpicture}
    \caption{A \tba for the language of the formula~$ \initsymb \land X G\neg\initsymb \land G_{[0,1]}\neg b \wedge G(a
      \rightarrow G_{[0,10]} \neg b)$ with accepting locations $q_1$ and $q_2$.} 
    \label{fig:aut-example2}
\end{figure}

\begin{example}\label{ex:partial-obs}
  Let us consider again the property ``$\initsymb\land F_{[0,10]}a \land
  G_{[0,20]}\neg b$'' but now we observe ``$a$'' (possibly) with 
  uncertainty on the timestamps and ``$b$'' is unobservable.

  For example, after the initial ``$\initsymb$'', we observe  ``$a$'' at time
  $0$, another time in the interval $[6,7]$ and a final time in the interval $[15,16]$, but no other $\initsymb$, and now is time~$30$.
  The words that are consistent with these observations have the form
  $(\initsymb,0)
  \rho_0 (a,0) \rho_1 (a, t_1) \rho_2 (a, t_2) \rho_3
  $ 
  where 
  \begin{itemize}
      \item $t_1 \in [6,7]$ and $t_2 \in [15,16]$,
      \item $\rho_0$ is a (possibly empty) finite timed word~$(b,0) \cdots (b,0)$,
      \item $\rho_1$ is a (possibly empty) finite timed word~$(b,t_{1,1}) \cdots (b,t_{1,n_1})$ with $t_{1,j} \in [0,t_1]$ for all $1 \le j \le n_1$,
      \item $\rho_2$ is a (possibly empty) finite timed word~$(b,t_{2,1}) \cdots (b,t_{2,n_2})$ with $t_{2,j} \in [t_1,t_2]$ for all $1 \le j \le n_2$, and
      \item $\rho_3$ is a (possibly empty) finite timed word~$(b,t_{3,1}) \cdots (b,t_{3,n_3})$ with $t_{3,j} \in [t_2,30]$ for all $1 \le j \le n_3$.
  \end{itemize}

  Without assumptions we cannot have any conclusive verdict, because
  we do not know if a ``$b$'' occurred before time point~$20$ or not. But with the assumption from
  the previous example, we can conclude at time $16$ that the
  property is true:
\begin{itemize}
    \item $\rho_0$ must be empty, as there cannot be a ``$b$'' within the first unit of time.
    \item $\rho_1$ must be empty, as there cannot be a ``$b$'' for ten units of time after the ``$a$'' at time point~$0$ and $t_1 \le 7 \le 10$.
    \item $\rho_2$ must be empty, as there cannot be a ``$b$'' for ten units of time after the ``$a$'' at time point~$t_1$ and $t_2 \le 16 \le t_1+10$.
    \item $\rho_3$ cannot contain a ``$b$'' with timestamp~$t_{3,j} \le 20$, as this would imply that a ``$b$'' has occurred less than ten units of time after the ``$a$'' at $t_2$.
\end{itemize}
  Thus, under the assumption, we can make a definitive verdict, which we could not without the assumption.
\end{example}

\subsection{Abstract Definition and Basic Properties}\label{sec:abstract}
In the following, we formalize this intuition. To develop the theory as general as possible, we allow real time points in the observations. Later, when we are concerned with algorithms, we will restrict ourselves to rational inputs.
In the same spirit, we begin with a very abstract definition of monitoring under assumptions. Later, we will explain how to represent the property, the assumption, and the observation finitely.

\begin{definition}
\label{def:assumpmonitor}
Let $\assumpverdicts = \{\top,\bot,\unknown,\outofmodel\}$.
Given a property~$\varphi \subseteq \TSigma^\omega$ of infinite timed words, an assumption~$A \subseteq \TSigma^\omega$, an  observation~$O \subseteq \TSigma^*$, and a current time point~$t \ge \tau(O)$, the function $\assumpevalfuncsymbol \colon (\pow{\TSigma^\omega} \times \pow{\TSigma^\omega}) \rightarrow (\pow{\TSigma^*}  \times \nnreals) \rightarrow \assumpverdicts$ evaluates to a verdict with the following definition:
\[
\assumpevalfunc{\varphi,A}{O,t} = \left.
  \begin{cases}
    \outofmodel & \text{if } O \cdot_t \TSigma^\omega \cap A = \emptyset, \\
    \top & \text{if } O \cdot_t \TSigma^\omega \cap A \neq \emptyset \text{ and } O \cdot_t \TSigma^\omega \cap A \subseteq \varphi, \\
    \bot & \text{if } O \cdot_t \TSigma^\omega \cap A \neq \emptyset \text{ and } O \cdot_t \TSigma^\omega \cap A \subseteq \TSigma^\omega \setminus \varphi, \\
    \unknown & \text{otherwise}.
  \end{cases}
  \right.
\]
$\assumpevalfunc{\varphi,A}{O,t}$ is undefined when $t < \tau(O)$.
\end{definition}

Before we study how to specify assumptions and observations, and how to compute $\assumpevalfuncsymbol$, we study some properties of $\assumpevalfuncsymbol$.
First of all, our definition satisfies the two maxims of \myquot{impartiality} and \myquot{anticipation}~\cite{DLT}:
Impartiality requires that if an observation yields a definitive verdict (i.e., $\top$, $\bot$, or $\outofmodel$ in our case), then every continuation of the observation also yields the same verdict. Note that impartiality is satisfied by always giving the verdict~$\unknown$.
Hence, anticipation requires that a definitive verdict is given as soon as possible.

To formalize this in the setting of inexact observations, we need to define what it means for an observation to extend another one.
Let $O,O' \subseteq \TSigma^*$ and $t \ge \tau(O)$. We say that $O'$ extends $O$ at $t$,
written $O \sqsubseteq_t O'$ if 
\begin{itemize}
    \item for every $\obs \in O$ there exists an $\obs^* \in \TSigma^*$ such that $\obs \cdot_t \obs^* \in O'$ and
    \item for every $\obs' \in O'$ there exists an $\obs \in O$ and an $\obs^* \in \TSigma^*$ such that $\obs' = \obs \cdot_t \obs^*$.
\end{itemize}

Also note that extending an observation that is consistent with an assumption~$A$ might lead to an observation that is no longer consistent with $A$, which implies that a (definitive) verdict~$\top$ or $\bot$ may turn into the definitive verdict~$\outofmodel$ when an observation is extended.
This has to be taken into account when defining impartiality for monitoring under assumptions.

\begin{lemma}
\label{lemma_impanti}
Fix a property~$\varphi \subseteq \TSigma^\omega$ and an assumption~$A \subseteq \TSigma^\omega$.
\begin{enumerate}
    \item \label{lemma_impanti_imp} $\assumpevalfuncsymbol$ is impartial, i.e., $\assumpevalfunc{\varphi, A}{O,t} \in \{\top, \bot, \outofmodel\}$ implies $\assumpevalfunc{\varphi, A}{O',t'} \in \{ \assumpevalfunc{\varphi, A}{O,t},\outofmodel\}$ for all $O \sqsubseteq_t O'$, $t \ge \tau(O)$, and $t' \ge \tau(O')$, i.e., a definitive verdict is stable under further observations, as long as the observations are consistent with the assumption.
    
    \item \label{lemma_impanti_anti}  $\assumpevalfuncsymbol$ is anticipatory, i.e., $\assumpevalfunc{\varphi, A}{O,t} = \unknown$ implies that 
    $O \cdot_t \TSigma^\omega \cap A \cap \varphi$ and $O \cdot_t \TSigma^\omega \cap A \cap \TSigma\setminus \varphi$ are both nonempty, i.e., the observation~$O$ is consistent with the assumption and satisfaction of $\varphi$ and the observation~$O$ is consistent with the assumption and the violation of $\varphi$.
\end{enumerate}
\end{lemma}

\begin{proof}
\ref{lemma_impanti_imp}.) If $O \sqsubseteq_t O'$, $t \ge \tau(O)$, and $t' \ge \tau(O')$, then $O \cdot_t \TSigma^\omega \supseteq O' \cdot_{t'} \TSigma^\omega$.
We will apply this fact in each of the cases of the following case distinction.

\begin{itemize}
    \item If $\assumpevalfunc{\varphi, A}{O,t} = \outofmodel$, then we have $O \cdot_t \TSigma^\omega \cap A = \emptyset$ by definition. 
    Hence, we obtain
\[
O' \cdot_{t'} \TSigma^\omega \cap A \subseteq O \cdot_t \TSigma^\omega \cap A = \emptyset,
\]
which implies $\assumpevalfunc{\varphi, A}{O',t'} = \outofmodel$, as required.

\item If $\assumpevalfunc{\varphi, A}{O,t} = \top$, then we have $O \cdot_t \TSigma^\omega \cap A \neq \emptyset$ and  $O \cdot_t \TSigma^\omega \cap A \subseteq \varphi$ by definition.
Now, if $O' \cdot_{t'} \TSigma^\omega \cap A = \emptyset$, then we obtain $\assumpevalfunc{\varphi, A}{O',t'} = \outofmodel$, i.e., the observation~$O$ was extended to $O'$ in a way that is not longer consistent with the assumption~$A$.
On the other hand, if $O' \cdot_{t'} \TSigma^\omega \cap A \neq \emptyset$, then
\[
O' \cdot_{t'} \TSigma^\omega \cap A \subseteq O \cdot_t \TSigma^\omega \cap A \subseteq \varphi,
\]
which implies $\assumpevalfunc{\varphi, A}{O',t'} = \top$, as required.

\item The reasoning for $\assumpevalfunc{\varphi,A}{O,t} =  \bot$ is dual, we just have to replace $\top$ by $\bot$ and $\varphi$ by $\TSigma^\omega \setminus \varphi$ in the previous case.
\end{itemize}

\ref{lemma_impanti_anti}.) This follows directly from the definition.
\end{proof}

Now, let us order $\assumpverdicts$ with the partial order shown in Figure~\ref{fig_order}, with the intuition that $\verdle$ orders verdicts by their specificity. 

\begin{figure}[h]
    \centering
    \begin{tikzpicture}
        \node (unk) at (0,0) {$\unknown$};
        \node (bot) at (-1,1) {$\bot$};
        \node (top) at (1,1) {$\top$};
        \node (oom) at (0,2) {$\outofmodel$};

        \node[rotate=45] at (.5,.5) {$\verdle$};
        \node[rotate=135] at (-.5,.5) {$\verdle$};
        \node[rotate=135] at (.5,1.5) {$\verdle$};
        \node[rotate=45] at (-.5,1.5) {$\verdle$};

    \end{tikzpicture}
    \caption{The (partial) specificity order~$\verdle$ on $\assumpverdicts$.}
    \label{fig_order}
\end{figure}

Then, more restrictive assumptions and more precise observations lead to more specific verdicts.

\begin{lemma}
Let $A \supseteq A'$, $O \supseteq O'$, and $t \le t'$ with $t \ge \tau(O)$ and $t' \ge \tau(O')$. Then,
\[
\assumpevalfunc{\varphi,A}{O,t} \verdle \assumpevalfunc{\varphi,A'}{O',t'}
.\]
\end{lemma}

\begin{proof}
We proceed by case distinction.

\begin{itemize}
    \item If $\assumpevalfunc{\varphi,A}{O,t} = \unknown $, the claim is vacuously true, as $\unknown$ is the $\verdle$-smallest element in $\assumpverdicts$.

    \item If $\assumpevalfunc{\varphi,A}{O,t} = \top$, then we have $O \cdot_t \TSigma^\omega \cap A \neq \emptyset $ and $O \cdot_t \TSigma^\omega \cap A \subseteq \varphi$ by definition. Now, if we have $O' \cdot_{t'} \TSigma^\omega \cap A' = \emptyset $, then $\assumpevalfunc{\varphi,A'}{O',t'} = \outofmodel$, which, as $\verdle$-largest element, is more specific than $\assumpevalfunc{\varphi,A}{O,t}$. 
    On the other hand, if $O' \cdot_{t'} \TSigma^\omega \cap A' \neq \emptyset $, then applying Remark~\ref{remark_concatprops} yields
    \[
    O' \cdot_{t'} \TSigma^\omega \cap A' \subseteq O \cdot_t \TSigma^\omega \cap A \subseteq \varphi.
    \]
    Thus,
    \[
    \assumpevalfunc{\varphi,A'}{O',t'} = \top \verdge \top = \assumpevalfunc{\varphi,A}{O,t}
    \]
    as required.

    \item The reasoning for $\assumpevalfunc{\varphi,A}{O,t} =  \bot$ is dual, we just have to replace $\top$ by $\bot$ and $\varphi$ by $\TSigma^\omega \setminus \varphi$ in the previous case.

    \item Finally, if $\assumpevalfunc{\varphi,A}{O,t} = \outofmodel$, i.e., we have $O \cdot_t \TSigma^\omega \cap A = \emptyset$, then applying Remark~\ref{remark_concatprops} yields
    \[
    O' \cdot_{t'} \TSigma^\omega \cap A' \subseteq O \cdot_t \TSigma^\omega \cap A = \emptyset,
    \]
    which implies
    \[
    \assumpevalfunc{\varphi,A'}{O',t'} = \outofmodel \verdge \outofmodel = \assumpevalfunc{\varphi,A}{O,t}
    \]
    as required.\qedhere
\end{itemize}
\end{proof}

\subsection{From the Abstract Definition to a Concrete Definition}\label{sec:concrete}

Our goal is to present an algorithm computing the monitoring function~$\assumpevalfuncsymbol$ in the concrete setting where
\begin{itemize}
    \item the property~$\varphi$ and its complement are accepted by \tba's (this covers in particular the case of $\varphi$ being given in \mitl due to Theorem~\ref{thm:mtltba}),
    \item the assumption~$A$ is given by a \tba, and
    \item the observation~$O$ is given by a sequence of time-intervals and propositional formulas over the locations, the clock constraints, and the alphabet of the assumption automaton capturing the (incomplete) information observed during monitoring.
\end{itemize}

We begin by introducing the assumption and observations. 
The former is given by a \tba, which we typically denote by~$\auta$ to distinguish it from other \tba.
Thus, let $\auta = (Q, Q_0, \Sigma, C, \Delta, \mathcal{F})$ be a \tba, i.e., $Q$ is the set of locations, $\Sigma$ is the alphabet, and $C$ is the set of clocks.
Let $G'(C)$ denote the clock constraints over $C$ with non-negative rational bounds, i.e., conjunctions of atomic constraints of the form~$c \sim t$, where $c \in C$ is a clock, $t \in \nnrats$, and $\sim\ \in \{<, \leq, =, \geq, >\}$.
Let $\phi$ be a (finite) propositional formula over the set~$ \Sigma \cup Q \cup G'(C)$ of propositions (which is infinite!), let $\sigma \in \Sigma$, and let $(q,v)$ be a state of $\auta$.
We define $\sigma, (q, v) \models \phi$ as follows:
\begin{itemize}
     \item For $\sigma' \in \Sigma$, $\sigma, (q, v) \models \sigma'$ if and only if $\sigma' = \sigma$.
     \item For $q' \in Q$, $\sigma, (q, v) \models q'$ if and only if $q' = q$.
     \item For $g \in G'(C)$, $\sigma, (q, v) \models g$ if and only if $g$ is satisfied by $v$.
     \item The semantics of Boolean connectives is defined as usual.
\end{itemize}

An $\auta$-observation is a finite sequence~$\obs = (\phi_1,I_1,m_1)\cdots (\phi_n,I_n,m_n)$ where the $\phi_j$ are (finite) propositional formulas over $ \Sigma \cup Q \cup G'(C)$, the $I_j$ are bounded intervals of $\nnreals$ with rational endpoints (which may overlap), and the multiplicities~$m_j$ are in $\{\le \multbound, = \multbound, \ge \multbound \mid \multbound \in \nats_0\}$. 
These multiplicities can describe bounds on the number of occurrences of events. For example, 
\begin{itemize}
    \item exactly \(k\) events that occur in the interval $I$ and satisfy $\phi$ is captured by the triple~(\(\phi, I, =k\)), 
    \item an unbounded number of events that occur in the interval $I$ and satisfy $\phi$ is captured by (\(\phi, I, \ge 0\)), 
    \item at most \(k\) events that occur in the interval $I$ and satisfy $\phi$ is captured by (\(\phi, I, \le k\)), 
    \item at least \(k\) events that occur in the interval $I$ and satisfy $\phi$ is captured by (\(\phi, I, \ge k\)), and 
    \item a number of events in \([\ell, u]\) that occur in the interval $I$ and satisfy $\phi$ is captured by the concatenation~(\(\phi, I, = \ell\))(\(\phi, I, \le u-\ell\)).
\end{itemize}

Let $\rho = (\sigma_1, \tau_1) \cdots (\sigma_{n'}, \tau_{n'})$ be a finite timed word and let 
\[r = (q_0, v_0) \transition{1} (q_1,v_1) \transition{2} \cdots \transition{n'-1} (q_{n'-1}, v_{n'-1}) \transition{n'} (q_{n'},v_{n'}) \]
be a prefix of a run of $\auta$ with $q_0 \in Q_0$ and $v_0(c) = 0$ for all $c \in C$ (note that $r$ processes $\rho$). 
Then, we say that $r$ witnesses that $\rho$ is consistent with an observation~$\obs = (\phi_1,I_1,m_1)\cdots (\phi_n,I_n,m_n)$, if there is a function~$h\colon \{1,2,\ldots, n'\} \rightarrow \{1,2,\ldots, n\}$ such that
\begin{enumerate}
    \item $h(1) \le h(2) \le \cdots \le h(n')$,
    \item $\sigma_{j'}, (q_{j'}, v_{j'}) \models \phi_{h(j')}$ for all $j' \in \{1,2,\ldots, n'\}$,
    \item $\tau_{j'} \in I_{h(j')}$ for all $j' \in \{1,2,\ldots, n'\}$, and
    \item for all $j \in \{1,2,\ldots, n\}$:
    \begin{itemize}
        \item If $m_j$ is equal to ${\le \multbound}$, then $\size{\{j' \in \{1,2,\ldots, n'\} \mid h(j') = j\}} \le \multbound$.
        \item If $m_j$ is equal to ${= \multbound}$, then $\size{\{j' \in \{1,2,\ldots, n'\} \mid h(j') = j\}} = \multbound$.
        \item If $m_j$ is equal to $\ge \multbound$, then $\size{\{j' \in \{1,2,\ldots, n'\} \mid h(j') = j\}} \ge \multbound$.
    \end{itemize}
\end{enumerate}
Intuitively, the function \(h\) maps indexes of $r$ and $\rho$ to indexes of $o$. Each index of $r$ and $\rho$ has to satisfy an element in $o$ (i.e., the one it is mapped to by $h$), and each element of $o$ has to be satisfied by a number of elements in $r$ and $\rho$ according to its multiplicity (i.e., the elements that are mapped to it by $h$).

\begin{definition}   
The language~$\conswords{\auta}{ \obs } \subseteq \TSigma^*$ contains all words that are  consistent with~$\obs$.
\end{definition}

\begin{example}\label{ex:observation}
Let us continue Example~\ref{ex:partial-obs} and let $\auta$ be
the assumption automaton shown in
Figure~\ref{fig:aut-example2}. Consider the $\auta$-observation
\begin{multline*}
\obs=(\initsymb,[0,0],=1)(a,[0,0],=1)(\neg a,[0,7],\ge 0)(a,[6,7],=1)\\(\neg a,[6,16],\ge 0)(a,[15,16],=1)(\neg a,[0,30],\ge 0).
\end{multline*}
Then, as argued in Example~\ref{ex:partial-obs}, $\conswords{\auta}{ \obs }$ is the language 
\begin{multline*}
\{
(\initsymb,0)(a,0) (a, t_1) (a, t_2) (b, t_{3,1}) \cdots (b, t_{3,n_3}) \mid\\ t_1 \in [6,7], t_2 \in [15,16]\text{, and } t_2 +10 < t_{3,1} \le \cdots \le t_{3,n_3} \le 30
\}    .
\end{multline*}

For example, given the run prefix (we ignore the clock~$x$ as it is always equal to the timestamp on the transition leading to a state)
\[r_0 = (q_0, y=0) \labeledtransition{(\initsymb,0)}  (q_1, y=0) \labeledtransition{(a,0)} (q_2,y=0) \labeledtransition{(a,6)} (q_{2}, y=0) \labeledtransition{(a,15)} (q_{2},y=0) \]
we can define $h$ as follows: $h(1)=1, h(2)=2,
h(3)=4$, $h(4)=6$. 
Thus, $r_0$ witnesses that $(\initsymb,0)(a,0)(a,6)(a,15)$ is in $\conswords{\auta}{\obs}$.

For the run prefix
\[r_1 = (q_0, y=0) \labeledtransition{(\initsymb,0)}(q_1, y=0) \labeledtransition{(a,0)} (q_2,y=0) \labeledtransition{(a,6)} (q_{2}, y=0) \labeledtransition{(b,15)} (q_{3},y=9) \labeledtransition{(a,16)} (q_{3},y=10)\]
we can define the function $h$ as follows: $h(1) =1, h(2)=2, h(3)=4,
h(4)=5, h(5)=6$.
Thus, $r_1$ witnesses that $(\initsymb,0)(a,0)(a,6)(b,15)(a,16)$ is in $\conswords{\auta}{\obs}$.
Note that $(\initsymb,0)(a,0)(a,6)(b,15)(a,16)$ is \emph{not} a prefix of any word that satisfies the assumption, as it contains an ``$a$'' that is followed by a ``$b$'' within nine units of time. 
But it is in $\conswords{\auta}{\obs}$, as membership in that set only depends on the existence of a run prefix, and not the existence of a prefix of an accepting run. 
The acceptance condition will be later taken care of (concretely in the first step of our algorithm in Subsection~\ref{sec:algorithm}).

Finally, the run prefix
\[r_2 = (q_0, y=0) \labeledtransition{(\initsymb,0)} (q_1, y=0) \labeledtransition{(a,0)} (q_2,y=0) \labeledtransition{(a,6)} (q_{2}, y=0) \labeledtransition{(a,15)} (q_{2},y=0) \labeledtransition{(a,16)} (q_{2},y=0)\]
is the prefix of a run of $\auta$ but it is not
compatible with the observation $\obs$. In fact, any $h$ satisfying
the conditions 1), 2), and 3) should assign $h(4)=6$ and $h(5)=6$
violating condition 4).
Thus, $r_2$ is not a witness for any word in $\conswords{\auta}{\obs}$.
\end{example}

Before we continue, let us collect some useful facts about $\auta$-observations.

\begin{remark}
We can restrict ourselves w.l.o.g.\ to multiplicities~$m_i$ from the set~$\{\le \multbound, = \multbound, \ge 0 \mid \multbound \in \nats\}$, as elements~$(\phi_i,I_i,m_i)$ with multiplicity~$m_i$ of the form~$\le 0$ or $= 0$ can be removed from the observation without changing the set of consistent words, and as each element~$(\phi_i,I_i,\ge \multbound)$, for some $\multbound > 0$, can be replaced by $(\phi_i,I_i,=\multbound)(\phi_i,I_i,\ge 0)$ without changing the set of consistent words.
Thus, in all following proofs we only consider such restricted multiplicities.
\end{remark}

Furthermore, let us note that there is a syntactic upper bound on the duration of finite timed words that can be consistent with an observation.

\begin{remark}
We have $\tau(\conswords{\auta}{ (\phi_1,I_1,m_1)\cdots(\phi_{n},I_{n},m_n)}) \le \sup I_n$ by definition.
\end{remark}

Before we study the algorithmic properties of our monitoring framework we have introduced it is instructive to compare it to the abstract definition of monitoring under assumptions presented in Definition~\ref{def:assumpmonitor}: 
In our concrete setting, the language~$L(\auta)$ takes the role of the assumption~$A$, the language~$\conswords{\auta}{\obs}$ takes the role of the observation~$O$, and the property~$\varphi$ is represented by two automata, one accepting $\varphi$ and one accepting its complement~$\TSigma^\omega \setminus \varphi$.
Note that here the observation depends on the \tba inducing the assumption, as observations can refer to locations and clock valuations of $\auta$.

Next, we show that we can determine whether an observation has at least one finite word that is consistent with the assumption. If this is not the case, i.e., $\conswords{\auta}{\obs}$ is empty, then $\obs$ is not consistent with the assumption~$\auta$.
However, the converse is in general not true, i.e., $\conswords{\auta}{\obs}$ may be nonempty, but the words in it cannot be extended to infinite words in $L(\auta)$.

\begin{lemma}
\label{lemma_obsemptinesstest}
The problem \myquot{Given a \tba~$\auta$ and an $\auta$-observation~$\obs$, is $\conswords{\auta}{\obs}$ nonempty?} is decidable.
\end{lemma}

\begin{proof}
We inductively (over the length of $\obs$) construct a timed automaton~$\auta \otimes\obs$ (over finite timed words) that accepts $\conswords{\auta}{\obs}$. 
Our result follows then from non-emptiness of timed automata being in PSPACE~\cite{alur1994tba}.

Let $\auta = (Q, Q_0, \Sigma, C, \Delta, \mathcal{F})$.
The timed automata~$\auta \otimes\obs$ we construct satisfy the following invariants: The set of locations~$Q'$ of $\auta \otimes\obs$ is always a subset of $Q \times S$ for some finite set~$S$ (but not necessarily equal to $Q \times S$).
In the inductive proof, we assume that $S \cap \nats$ is empty, which can always be achieved by renaming locations. This allows us to use subsets of $Q \times \nats$ as fresh states during the induction.
Further, the set of clocks of each $\auta \otimes\obs$ is equal to $C \cup \{time\}$, where $time$ is a fresh clock that is never reset, i.e., it measures the time that has passed since the start of a run.

The automata~$\auta \otimes\obs$ we construct have rationals in their clock constraints. 
These can be eliminated by multiplying each one of them by the largest denominator appearing in the automaton, which results in an emptiness-equivalent timed automaton.
Furthermore, we use clock constraints of the form~$time \in I$, for $I = [\ell, u]$ with rational endpoints~$\ell$ and $u$, as shorthands for $time \ge \ell \wedge time \le u$.

Now, we can begin the inductive construction by considering the empty observation~$\obs = \varepsilon$: We define $\auta\otimes\varepsilon = (Q \times \{s\}, Q_0 \times\{s\}, \Sigma, C \cup \{time\}, \emptyset, Q_0\times\{s\} )$, where $s$ is a fresh symbol.
This timed automaton satisfies the invariants and accepts the language~$\{\varepsilon\}$ if $Q_0$ is nonempty, otherwise it accepts the empty language. This is in both cases equal to $\conswords{\auta}{\varepsilon}$.

For the induction step, consider an observation $\obs(\phi, I, m)$, i.e., we concatenate $\obs$ with $(\phi, I, m)$.
Let $q \in Q$ and $a \in \Sigma$, and let $\phi_{(q,a)}$ be the formula obtained from $\phi$ by replacing each subformula~$a$ and $q$ by $\true$ and each subformula~$a' \neq a$ and $q' \neq q$ by $\false$, i.e., the only remaining atomic subformulas are clock constraints.
By bringing $\phi_{(q,a)}$ into disjunctive normal form and exploiting that clock constraints are closed under conjunction, we obtain that $\phi_{(q,a)}$ is equivalent to a formula of the form~$\bigvee_{d \in D_{(q,a)}} g_{(q,a,d)}$ for some finite index set~$D_{(q,a)}$, where each $g_{(q,a,d)}$ is a clock constraint.

By induction hypothesis, there is a timed automaton~$\auta\otimes \obs = (Q', Q_0', \Sigma, C\cup\{time\}, \Delta', \mathcal{F}')$ with language~$\conswords{\auta}{\obs}$.
We construct $\auta\otimes (\obs(\phi, I, m))$ by case distinction over the form of $m$ by intuitively extending $\auta\otimes \obs$ by transitions capturing $(\phi, I, m)$.

First, assume $m$ is of the form~\myquot{$= e$} for some $e >0$. Intuitively, we extend $\auta\otimes\obs$ by runs of $\auta$ of length exactly $e$ that additionally have to satisfy the requirements spelled out by $\phi$ and $I$, which can be added as clock constraints on the edges.
    To this end, we extend $\auta\otimes\obs$ by adding $e$ copies of $\auta$'s locations and edges leading from $\mathcal{F}'$ to the first copy as well as edges leading from the $j$-th copy to the $(j+1)$-th (see Figure~\ref{fig_constr_equals} for an illustration).
    
    \begin{figure}[h]
    \centering
    \begin{tikzpicture}[thick]

\draw [
    decoration={
        brace,
        mirror,
        amplitude=7pt
    },
    decorate
] (2,-1) -- (9,-1); 

\node[] at (5.5,-1.5) {$e$ copies};

\node at (0,0) {$\auta\otimes\obs$};
\node at (2.5,0) {$\auta$};
\node at (4.5,0) {$\auta$};
\node at (8.5,0) {$\auta$};

  \draw[] (0,0) ellipse (.75 and 1);
  \draw[] (2.5,0) ellipse (.5 and .75);
  \draw[] (4.5,0) ellipse (.5 and .75);
  \node at (6.5,0) {$\cdots$};
  \draw[] (8.5,0) ellipse (.5 and .75);

\path[-stealth, ultra thick] (.6,0) edge[bend left=10] (2.3,0);
\path[-stealth, ultra thick] (2.85,0) edge[bend left=10] (4.3,0);  
\path[-stealth, ultra thick] (4.85,0) edge[bend left=10] (5.9,0);  
\path[-stealth, ultra thick] (7.1,0) edge[bend left=10] (8.3,0);  
    \end{tikzpicture}
    \caption{Extending $\auta\otimes\obs$ to obtain $\auta\otimes(\obs(\phi, I, = e))$.}
    \label{fig_constr_equals}
\end{figure}
    
    Thus, we define $\auta\otimes (o(\phi, I, m)) = (Q'', Q_0', \Sigma, C \cup \{time\}, \Delta'', \mathcal{F}'' )$ where 
    \begin{itemize}
        \item $Q'' = Q' \cup (Q \times \{1,2,\ldots, e\})$ (recall that we assume w.l.o.g., that $Q' \subseteq Q \times S$ with $S \cap \nats = \emptyset$, which implies that the union is disjoint), 
        \item $\Delta''$ contains
    \begin{itemize}
        \item all transitions from $\Delta'$,
        \item all transitions of the form~$( (q,s),(q',1),a,\lambda,g \wedge g_{(q',a,d)} \wedge time \in I )$ for $(q,q',a,\lambda,g) \in \Delta$, $d \in D_{(q,a)}$, and $(q,s) \in \mathcal{F}'$, as well as
        \item all transitions of the form~$( (q,j),(q',j+1),a,\lambda,g \wedge g_{(q',a,d)} \wedge time \in I )$ for $(q,q',a,\lambda,g) \in \Delta$, $d \in D_{(q,a)}$, and $1 \le j < e$,     and
    \end{itemize}
    \item
    $\mathcal{F}''$ is equal to $Q \times \{e\}$.
    \end{itemize}

    Now, assume $m$ is of the form~\myquot{$\le e$} for some $e >0$.
    Here, we use the same transition structure as in the previous case, but runs do not have to reach~$Q \times \{e\}$ to be accepting (capturing \myquot{exactly $e$ matches of $(\phi, I, m)$}), but can stop early. Formally, we define $\auta\otimes (o(\phi, I, m)) = (Q'', Q_0', \Sigma, C \cup \{time\}, \Delta'', \mathcal{F}'' )$ where $Q''$, and $\Delta''$ are as in the previous case, but $\mathcal{F}''$ is equal to $\mathcal{F'} \cup (Q \times \{1,2,\ldots, e\})$, which captures \myquot{at most $e$ matches of $(\phi, I, m)$}.

% \begin{figure}[h]
%     \centering
%       \begin{tikzpicture}[thick]

% \node at (0,0) {$\auta\otimes\obs$};
% \node at (2.5,0) {$\auta$};
% \node at (4.5,0) {$\auta$};
% \node at (8.5,0) {$\auta$};

% \draw [
%     decoration={
%         brace,
%         mirror,
%         amplitude=7pt
%     },
%     decorate
% ] (2,-1) -- (9,-1); 

% \node[] at (5.5,-1.5) {$e$ copies};

%   \draw[] (0,0) ellipse (.75 and 1);
%   \draw[] (2.5,0) ellipse (.5 and .75);
%   \draw[] (4.5,0) ellipse (.5 and .75);
%   \node at (6.5,0) {$\cdots$};
%   \draw[] (8.5,0) ellipse (.5 and .75);

% \path[-stealth, ultra thick] (.85,0) edge[bend left=10] (1.9,0);
% \path[-stealth, ultra thick] (3.1,0) edge[bend left=10] (3.9,0);  
% \path[-stealth, ultra thick] (5.1,0) edge[bend left=10] (5.9,0);  
% \path[-stealth, ultra thick] (7.1,0) edge[bend left=10] (7.9,0);  

% \path[-stealth, ultra thick] (.6,.75) edge[bend left=10] (4.25,.75);
% \path[-stealth, ultra thick] (.6,.75) edge[bend left=15] (6.25,.75);
% \path[-stealth, ultra thick] (.6,.75) edge[bend left=20] (8.25,.75);

%     \end{tikzpicture}
%     \caption{Extending $\auta\otimes\obs$ to obtain $\auta\otimes(\obs(\phi, I, \le e))$.}
%     \label{fig_constr_le}
% \end{figure}

    Finally, assume $m$ is of the form~\myquot{$\ge 0$}. Here, we intuitively extend $\auta\otimes\obs$ by runs of $\auta$ of arbitrary (potentially zero) length by intuitively concatenating  $\auta\otimes\obs$ (at the accepting locations) and a copy of $\auta$ with $\varepsilon$-transitions: from an accepting location of the form~$(q,s)$ of $\auta\otimes\obs$ an $\varepsilon$-edge leads to the copy of $q$ of $\auta$. In the copy, we again have to satisfy the requirements spelled out by $\phi$ and $I$. 
    
    Since we do not allow $\varepsilon$-transitions in timed automata, we instead define $\auta \otimes(\obs(\phi, I, m))$ to be the result of the classical construction removing such transitions: for each transition of $\auta$ leading from $q$ to $q'$ and each accepting~$(q,s)$ in $\auta\otimes\obs$, we add a transition from $(q,s)$ to the copy of $q'$.
    Every location in the copy is accepting. Furthermore, the accepting locations of $\auta\otimes\obs$ are also accepting in $\auta\otimes(\obs(\phi, I, m))$ to allow for zero matches of $(\phi, I, m)$). 

    Thus, we define $\auta\otimes (\obs(\phi, I, m)) = (Q'', Q_0'', \Sigma, C \cup \{time\}, \Delta'', \mathcal{F}'' )$ where 
    \begin{itemize}
    \item
    $Q'' = Q' \cup (Q \times \{1\})$ (recall that we assume w.l.o.g., that $Q' \subseteq Q \times S$ with $S \cap \nats = \emptyset$, which implies that the union is disjoint), 
    \item $Q_0'' = Q_0'$,
    \item $\Delta''$ contains
    \begin{itemize}
        \item all transitions from $\Delta'$,
        \item all transitions of the form~$( (q,s),(q',1),a,\lambda,g \wedge g_{(q',a,d)} \wedge time \in I)$ for $(q,q',a,\lambda,g) \in \Delta$ with $(q,s) \in \mathcal{F}'$, and $d \in D_{(q,a)}$ and
        \item all transitions of the form~$( (q,1), (q', 1), a, \lambda, g \wedge g_{(q',a,d)} \wedge time \in I)$ for $(q,q',a,\lambda,g) \in \Delta$ and $d \in D_{(q,a)}$, and
    \end{itemize}
    \item $\mathcal{F}''$ is equal to $\mathcal{F}' \cup (Q \times \{1\})$.
\end{itemize}

In all three cases, the invariants are satisfied after renaming the states. 
Now, an induction shows that the language of $\auta\otimes\obs$ is equal to $\conswords{\auta}{\obs}$ and another induction shows that the number locations of $\auta\otimes\obs$ is bounded by $n \cdot e^*$, where $n$ is the number of locations of $\auta$ and $e^*$ is the sum of the constants in the multiplicities (where each multiplicity of the form~$\ge 0$ contributes $1$ to the sum instead of $0$).
Note that $e^*$ is the number of copies of $\auta$ used to construct~$\auta \otimes\obs$.
Furthermore, $\auta\otimes\obs$ uses only one more clock than $\auta$.
However, the number of transitions may be exponential, as we turn arbitrary Boolean formulas (the $\phi$ in the observations) into disjunctive normal form, which can incur an exponential blowup.
The number of transitions of $\auta\times\obs$ in turn depends on the size of those formulas in disjunctive normal form.
Finally, note that the size of the bounds used in the guards of $\auta\times\obs$ depends on the bounds in $\auta$ and the bounds in $\obs$.
\end{proof}

\section{A Zone-Based Monitoring Algorithm}
\label{sec:impl}

In this section, we present an algorithm computing the monitoring function~$\assumpevalfuncsymbol$. To this end, we first need to introduce some notation for \tba and zones to represent subsets of states of \tba, which may be uncountable.

In Subsection~\ref{sec:zones} we introduce zones to symbolically represent sets of states of timed automata.
Then in Subsection~\ref{sec:characterizing} we present a characterization of the monitoring function using reachability sets in timed automata.
These reachability sets are shown to be computable in Subsection~\ref{sec:computing} using a zone-based algorithm.
This is the basis of our monitoring algorithm which is presented in Subsection~\ref{sec:algorithm}.

\subsection{Zones}\label{sec:zones}
For the monitoring algorithm, we use -- as is standard in analysing timed automata models -- symbolic states being pairs~$(q, Z)$ of locations and zones. A zone is a finite conjunction of constraints of the form~$x \sim t$ and $x - x' \sim t$ for clocks~$x,x'$, constants $t \in \nnrats$, and $\sim\ \in \{<, \leq, =, \geq, >\}$. We write $v \models Z$ to denote that the clock valuation $v$ satisfies all constraints of the zone $Z$. Furthermore, we interchangeably treat zones as conjunctions of clock constraints as well as sets of clock valuations that satisfy its constraints.
Given two zones $Z$ and $Z'$ over a set $C$ of clocks, and a set~$\lambda \subseteq C$ of clocks, we define the following operations on zones (which can be efficiently implemented using the DBM data-structure \cite{DBLP:conf/ac/BengtssonY03}):
\begin{itemize}
    %\item $\textit{free}_\lambda(Z) = \{v \mid \forall x \in C\text{ such that } v(x) \models Z \textit{ if } x \notin \lambda\}$
    %\item $Z_\textit{free} = \textit{free}_C(Z)$
    \item $Z[\lambda] = \{v \mid \exists v' \models Z \: \text{ such that }  v(x) = 0 \text{ if } x \in \lambda \text{, otherwise } v(x) = v'(x) \}$
    \item $Z^\nearrow = \{v \mid \exists v' \models Z \text{ such that }  v = v' + d  \text{ for some } d \in \nnreals\}$
    %\item $Z^\searrow = \{v \mid \exists v' \models Z \text{ such that }  v = v' - d \textit{ for some } d \in \nnreals\}$
    % \item $Z^{\nearrow_t} = \{v \mid \exists v' \in Z \text{ such that }  v = v' + t\}$
    %\item $Z^{\nearrow_I} = \{v \mid \exists v' \models Z \: \forall x \in C \text{ such that }  v'(x) + t_1 \le v(x) \le v'(x) + t_2\}$
    \item $Z \land Z' = \{v \mid v \models Z \textit{ and } v \models Z'\}$.
 %   \item $Z_{\downarrow D}= \{ v_{\downarrow D} \mid v\models Z\}$, 
% \,\,where $v_{\downarrow D}$ is the projection of $v$ onto $D$.
    % \item $Z_0 = Z[C]$.
    %\item $Z_\textit{free} = \{v \mid \forall x \in C\text{ such that }  0 \le v(x) < \infty \}$
\end{itemize}

To describe our algorithm, we first define the set of states of a \tba from where it is possible to reach an accepting location infinitely many times in the future, i.e., those states from which an accepting run is possible.
This is useful, because if processing a finite timed word leads to such a state, then the timed word can be extended to an infinite one in the language of the automaton, a notion that underlies Definition~\ref{def:assumpmonitor}.
Given a \tba $\aut = (Q, Q_0, \Sigma, C, \Delta, \mathcal{F})$, the set of states with nonempty language is
\[\nonempty{\aut} = \{ (q,v) \mid q \in Q, v \in C \rightarrow \nnreals \text{ such that } \lang{\aut, (q,v)} \neq \emptyset \}.\]

\begin{proposition}[\cite{GrosenKLZ22}]
$\nonempty{\aut}$ is a finite union of zones and can be computed using a zone-based algorithm.
\end{proposition}

\subsection{Characterizing Monitoring via Reachability Sets}\label{sec:characterizing}
We continue by capturing the set of states of a \tba that can be reached by processing a finite timed word.
Let $\aut = (Q, Q_0, \Sigma, C, \Delta, \mathcal{F})$ be a \tba.
In the following definition, we write $(q_0, v_0) \xrightarrow{\rho}_\aut (q_n, v_n)$ for a finite timed word~$\rho = (\sigma_1, \tau_1) \cdots (\sigma_n, \tau_n) \in \TSigma^*$ to denote the existence of a finite sequence
\[(q_0, v_0) \transition{1} (q_1,v_1) \transition{2} \cdots \transition{n} (q_n,v_n) \]
 of states,
where for all $1 \leq i \leq n$ there is a transition $(q_{i-1},q_{i},\sigma_{i},\lambda_i,g_i)$ such that $v_{i}(c) = 0$ for all $c$ in $\lambda_i$ and $v_{i-1}(c) + (\tau_i - \tau_{i-1})$ otherwise, and $g$ is satisfied by the valuation $v_{i-1}+(\tau_{i} - \tau_{i-1})$, where we use $\tau_0 = 0$. 
Given a TBA~$\aut$, a finite timed word~$\rho \in \TSigma^*$, and a time-point~$t \in \nnreals$ with $t \ge \tau(\rho)$, the set of possible states a run over $\rho$ starting from initial states of $\aut$ can end in after time~$t$ has passed is
\[
\terminal{\aut}{\rho,t} = \bigcup\nolimits_{q_0 \in Q_0} \{(q, v + (t-\tau(\rho))) \mid (q_0, v_0) \xrightarrow{\rho}_\aut (q, v) \},
\]
where $v_0$ is the clock valuation mapping every clock to $0$.
We call $\terminal{\aut}{\rho,t}$ the reach-set of $\aut$ over $(\rho, t)$.
The above definition is adapted from \cite{GrosenKLZ22} to take into account the time that has passed since the last observation, i.e., the value~$t - \tau(\rho)$. 

Next, we lift this definition to sets~$L \subseteq \TSigma^*$ of finite words via
\[
\terminal{\aut}{L,t} = \bigcup\nolimits_{\rho \in L} \terminal{\aut}{\rho,t},
\]
assuming $t \ge \tau(L)$. 
Otherwise, $\terminal{\aut}{L,t}= \emptyset$ by convention.

First, we show that reach-sets and states with a nonempty language allow us to characterize monitoring. Then, we show how to compute reach-sets of observations.

In the following lemma, $\conswords{\auta}{\obs}$ is the observation~$O$ in Definition~\ref{def:assumpmonitor} and $L(\auta)$ is the assumption~$A$.
Furthermore, the assumption on the property~$\varphi$ in the lemma is satisfied for all \mitl-definable properties.

\begin{lemma}
\label{lemma_moniviareachsets}
Let $\varphi \subseteq \TSigma^\omega$ be a property such that there are \tba's~$\aut_\varphi$ and $\aut_{\neg \varphi}$ with $L(\aut_\varphi) = \varphi$ and $L(\aut_{\neg \varphi}) = \TSigma^\omega \setminus \varphi$, let $\obs = (\phi_1,I_1,m_1)\cdots (\phi_n,I_n,m_n)$ be an $\auta$-observation for some \tba~$\auta$, and let $t \ge \sup I_n$.
\begin{enumerate}
    
    \item \label{lemma_moniviareachsets_assump} 
    $\conswords{\auta}{\obs} \cdot_t \TSigma^\omega \cap L(\auta) \neq \emptyset$
    if and only if $\terminal{\auta}{\conswords{\auta}{\obs},t} \cap \nonempty{\auta} \neq \emptyset$.
    
    \item \label{lemma_moniviareachsets_prop} $\conswords{\auta}{\obs} \cdot_t \TSigma^\omega \cap L(\auta)\subseteq \varphi$ if and only if $\terminal{\aut_{\neg \varphi} \otimes \auta}{\conswords{\auta}{\obs}, t} \cap \nonempty{\aut_{\neg \varphi} \otimes \auta} =\emptyset$.
   
    \item \label{lemma_moniviareachsets_negprop}  $\conswords{\auta}{\obs} \cdot_t \TSigma^\omega\cap L(\auta) \subseteq \TSigma^\omega \setminus \varphi$ if and only if $\terminal{\aut_{\varphi} \otimes \auta}{\conswords{\auta}{\obs}, t} \cap \nonempty{\aut_{ \varphi} \otimes \auta} =\emptyset$.
\end{enumerate}
\end{lemma}

\begin{proof}
\ref{lemma_moniviareachsets_assump}.)
For the left-to-right implication, let $\rho=(\sigma_1, \tau_1)(\sigma_2, \tau_2) \cdots \in \conswords{\auta}{\obs} \cdot_t \TSigma^\omega \cap L(\auta)$.
Due to $\rho \in \conswords{\auta}{\obs} \cdot_t \TSigma^\omega$, there is an $n\ge 0$ such that $\rho' = (\sigma_1, \tau_1) \cdots (\sigma_n, \tau_n) \in \conswords{\auta}{\obs}$, $\tau_n \le t$, and $\tau_{n+1} \ge t$.
On the other hand, due to $\rho \in L(\auta)$, there is an accepting run
\[(q_0, v_0) \transition{1} (q_1,v_1) \transition{2} (q_2,v_2) \transition{3} \cdots \]
of $\auta$.
Hence, we have $(q_0, v_0) \xrightarrow{\rho'}_\auta (q_n, v_n)$.
Thus, $(q_n, v_n + \delta) \in \terminal{\auta}{\rho',t} \subseteq \terminal{\auta}{\conswords{\auta}{\obs},t}$, where we define $\delta = t - \tau_n$.
It remains to show $(q_n, v_n + \delta) \in \nonempty{\auta}$.
To this end, consider the run suffix
\[(q_n, v_n) \transition{n+1} (q_{n+1},v_{n+2}) \transition{n+2} (q_{n+2},v_{n+2}) \transition{n+3} \cdots,\] which is accepting. 
Note that $\tau_{n} \le t \le \tau_{n+1}$ implies $\tau_{n+1} - \delta \ge \tau_n \ge 0$.
Hence, $\tau_{n+1} \le \tau_{n+2} \le \cdots$ implies $\tau_{n+j} - \delta \ge 0$ for all $j \ge 1$.
Thus, the run
\[(q_n, v_n + \delta) \transitiondeltam{n+1} (q_{n+1},v_{n+2}) \transitiondeltam{n+2} (q_{n+2},v_{n+2}) \transitiondeltam{n+3} \cdots \]
is well-defined.
As it starts in $(q_n, v_n+ \delta)$ and is accepting, we conclude $(q_n, v_n + \delta) \in \nonempty{\auta}$ as required.

For the right-to-left implication, let $(q, v) \in \terminal{\auta}{\conswords{\auta}{\obs},t} \cap \nonempty{\auta}$. 
Due to $(q, v) \in \terminal{\auta}{\conswords{\auta}{\obs},t}$ there is a $\rho' = (\sigma_1, \tau_1) \cdots (\sigma_n, \tau_n) \in \conswords{\auta}{\obs}$ with $ \tau_n \le t$ and there is a run prefix
\[(q_0, v_0) \transition{1} (q_1,v_1) \transition{2} \cdots \transition{n} (q_n,v_n) \]
for some $q_0 \in Q_0$, where $v_0$ is the clock valuation mapping every clock to zero, such that $v = v_n + \delta$, where we define $\delta = (t - \tau_n)$.
Furthermore, due to $(q, v) \in \nonempty{\auta}$, there is a $\rho = (\sigma_{n+1},\tau_{n+1})(\sigma_{n+2},\tau_{n+2})\cdots$ and an accepting run
\[(q_n', v_n') \transition{n+1} (q_{n+1}',v_{n+2}') \transition{n+2} (q_{n+2}',v_{n+2}') \transition{n+3} \cdots \]
starting in $(q,v)$ that processes $\rho$.

These two runs can be combined into the accepting run
\begin{align*}
&(q_0, v_0) \transition{1} (q_1,v_1) \transition{2} \cdots \transition{n} (q_n,v_n)\transition{n+1}\\
&(q_{n+1}',v_{n+2}') \transitiondelta{n+2} (q_{n+2}',v_{n+2}') \transitiondelta{n+3} \cdots
\end{align*}
processing $\rho' \cdot_t  \rho$, which is therefore in $\conswords{\auta}{\obs} \cdot_t \TSigma^\omega \cap L(\auta)$.

\ref{lemma_moniviareachsets_prop}.)
We show that the negations of both statements are equivalent. 

So, let $\conswords{\auta}{\obs} \cdot_t \TSigma^\omega \cap L(\auta)\not\subseteq \varphi$,
%$\conswords{\auta}{\obs} \cdot_t \TSigma^\omega \cap L(\auta)\cap (\TSigma^\omega \setminus \varphi) \neq \emptyset$
i.e., $\conswords{\auta}{\obs} \cdot_t \TSigma^\omega \cap L(\auta)\cap  (\TSigma^\omega\setminus \varphi) \neq \emptyset$.
Using arguments as in the left-to-right implication of Item~\ref{lemma_moniviareachsets_assump}., we can find a state~$(q,v) \in \terminal{\aut_{\neg \varphi} \otimes \auta}{\conswords{\auta}{\obs}, t} \cap \nonempty{\aut_{\neg \varphi} \otimes \auta}$, which is thus not empty.

Now, let $\terminal{\aut_{\neg \varphi} \otimes \auta}{\conswords{\auta}{\obs}, t} \cap \nonempty{\aut_{\neg \varphi} \otimes \auta}  \neq \emptyset$.
Using arguments as in the right-to-left implication of Item~\ref{lemma_moniviareachsets_assump}., we can find a word of the form~$\rho' \cdot_t \rho$ for $\rho' \in \conswords{\auta}{\obs}$ and $\rho \in \TSigma^\omega$ such that $\rho' \cdot_t \rho \in L(\auta) \cap L(\aut_{\neg\varphi})$.
Hence, $\rho' \cdot_t \rho$ is in $\conswords{\auta}{\obs} \cdot_t \TSigma^\omega \cap L(\auta)$, but not in $\varphi$.
Thus, $\conswords{\auta}{\obs} \cdot_t \TSigma^\omega \cap L(\auta)\not\subseteq \varphi$.

\ref{lemma_moniviareachsets_negprop}.) The reasoning is dual to Item~\ref{lemma_moniviareachsets_prop}, we just have to replace~$\varphi$ by $\TSigma^\omega \setminus \varphi$ and $\aut_{\neg\varphi}$ by $\aut_\varphi$.
\end{proof}

\subsection{Computing Reachability Sets}\label{sec:computing}
We now show how to compute reach-sets using zones.
First, we use the zone operations introduced above to compute the successor states of an input letter with a given target location. Fix a \tba $ (Q, Q_0, \Sigma, C, \Delta, \mathcal{F})$. For a symbolic state~$(q, Z)$, a letter $\sigma \in \Sigma$ and target location $q'\in Q$, we define \[\post((q, Z), \sigma, q') = \{ (q', Z') \mid (q,q',\sigma,\lambda,g) \in \Delta, Z' = (Z^\nearrow \land g)[\lambda] \},\] being the set of states one can reach by taking a $\sigma$-transition at some point in the future from $(q, Z)$ with $q'$ as target-location. 
Using $\post$ we can compute the successor states of a time-uncertain letter/target location $(\sigma, q', I)$, where $\sigma\in\Sigma$, $q'\in Q$  and $I\subseteq\nnreals$ is a time interval with rational endpoints.  For this, we extend zones with an additional clock $time$ just recording time since system start. The successors of a symbolic state $(q,Z)$ are 
\[\succc((q, Z), (\sigma, q', I)) = \{(q', Z') \mid (q', Z'') \!\in\! \post((q, Z), \sigma, q'), Z'\!\! =\! Z'' \land time\!\in\! I\}\] 
and the successors of a set~$S$ of symbolic states are
\[\succc(S, (\sigma, q', I)) = \bigcup\nolimits_{(q, Z)\in S} \succc((q, Z), (\sigma, q', I)).\]
Now, our main technical lemma below exploits the above to compute reach-sets.  More precisely, given an $\auta$-observation~$\obs$ (i.e., the \tba~$\auta$ represents the assumption), we can compute the reach-set of the set~$\conswords{\auta}{\obs}$ in the product~$\aut \otimes\auta$, for any given \tba~$\aut$, i.e., we compute the words consistent with the observation~$\obs$ in the \tba~$\auta$ (the assumption), while the reach-set of that language is computed in $\aut \otimes \auta$ (this will later be the product of the property (or its negation) and the assumption as in Lemma~\ref{lemma_moniviareachsets}).

\begin{lemma}
\label{lemma:reachset}
Fix \tba~$\auta, \aut$.
There is a zone-based online algorithm computing
\[ \terminal{\aut\otimes \auta}{\conswords{\auta}{(\phi_1,I_1,m_1)\cdots (\phi_n,I_n,m_n)},t}\] (which is a finite union of zones)
for every $\auta$-observation~$(\phi_1,I_1,m_1)\cdots (\phi_n,I_n,m_n)$, and every $t \in \nnrats$ with $t \ge \sup I_n$.
\end{lemma}

\begin{proof}
Recall that the state set of $\aut\otimes \auta$ has the form~$Q_\aut \times Q_\auta \times \{0,1\}$, where $Q_\aut$ and $Q_\auta$ are the sets of states of $\aut$ and $\auta$, respectively (see Proposition~\ref{prop:intersection}).

Let $\obs^i= (\phi_1,I_1,m_1)\cdots (\phi_i,I_i,m_i)$, and let us denote by 
$S^i_{\aut\otimes \auta}$  the reach-set $\terminal{\aut\otimes \auta}{\conswords{\auta}{\obs^i}}$ of $\obs^i$ in $\aut\otimes \auta$.
We will show inductively in $i$, that $S^i_{\aut\otimes \auta}$ can be computed using zone operations. 
For the base case $i=0$, we note that $\conswords{\auta}{\obs^0}=\{\varepsilon\}$, thus $\terminal{\aut\otimes \auta}{\conswords{\auta}{\obs^0}}$ is the set of initial states of $\aut\otimes\auta$, which is clearly effectively representable using zones. 

For the inductive case, let us assume that $S^{i-1}_{\aut\otimes \auta}$  is computable using zone operations. Now consider $(\phi_i,I_i,m_i)$. Given that $\Sigma$, $Q_\auta$ and $Q_\aut$ are finite, $\phi_i$ is equivalent to a finite disjunction of simple formulas of the form~$\psi_{i,j}=\sigma_{i,j}\wedge q^a_{i,j} \wedge g_{i,j}$, where $\sigma_{i,j}\in\Sigma$, $q^a_{i,j}\in Q_\auta$, and $g_{i,j}\in G'(C_\auta)$. 

Now in case $m_i$ is equal to $= \multbound$ for some $\multbound >0$, the reach-set with respect to  $\psi_{i,j}$ is simply $S^{i,\multbound}_{\aut\otimes \auta}$ where 
$S^{i,0}_{\aut\otimes \auta}$ is $S^{i-1}_{\aut\otimes \auta}$ and
\[S^{i,\multbound'}_{\aut\otimes \auta} = \bigcup\nolimits_j \bigcup\nolimits_{q^b \in Q_\aut}\bigcup\nolimits_{k \in \{0,1\}}
\succc(S^{i,\multbound'-1}_{\aut\otimes \auta}, (\sigma_{i,j},(q^b,q^a_{i,j},k), I_i))\wedge g_{i,j}\]
for $0 < \multbound' \le \multbound $.
Here, given a set~$S$ of symbolic states, we write $S\wedge g_{i,j}$ for the set of symbolic states of the form~$(q, Z \wedge q_{i,j})$ where $(q, Z)$ is in $S$.

In case $m_i$ is equal to $\ge 0$, $S^i_{\aut\otimes \auta}$ is the least fixed-point $\mathbf{X}$ satisfying the equality
\[
\mathbf{X} = S^{i-1}_{\aut\otimes \auta}\cup \bigcup\nolimits_j \bigcup\nolimits_{q^b \in Q_\aut}\bigcup\nolimits_{k \in \{0,1\}}\big[\succc(\mathbf{X}, (\sigma_{i,j},(q^b,q^a_{i,j},k), I_i))\wedge g_{i,j}\big].
\]
Given the upper bounds of the interval $I_i$, the least fixed-point will be found in a finite number of iterations of the right-hand-side of the above equation (starting from the empty set).

Finally, in case $m_i$ is equal to $\le \multbound$ for some $\multbound >0$, $S^i_{\aut\otimes \auta}$ is obtained by unraveling the fixed point above exactly~$\multbound$ times.

The above inductive proof  provides in an obvious manner the basis for an online construction of the sets $\terminal{\aut\otimes \auta}{\conswords{\auta}{\obs^i}}$.
\end{proof}

\begin{remark}
    For the special case of $\aut = \auta$ of Lemma~\ref{lemma:reachset}, we do not have to work on the product~$\auta \times \auta$, but can simplify the construction to work on $\auta$ directly. Thus, there is also a zone-based algorithm computing 
    \[ \terminal{\auta}{\conswords{\auta}{(\phi_1,I_1,m_1)\cdots (\phi_n,I_n,m_n)},t}\]
(which is again a finite union of zones)
for every $\auta$-observation~$(\phi_1,I_1,m_1)\cdots (\phi_n,I_n,m_n)$, and every $t \in \nnrats$ with $t \ge \sup I_n$.
\end{remark}

\subsection{Monitoring Algorithm}\label{sec:algorithm}
Now, we are able to present our algorithm to compute~$\assumpevalfuncsymbol(\varphi, A)$ for a property~$\varphi \subseteq \TSigma^\omega$ (given by two \tba~$\aut_\varphi$ and $\aut_{\neg\varphi}$ such that $L(\aut_\varphi) = \varphi$ and $L(\aut_{\neg \varphi}) = \TSigma^\omega \setminus \varphi$) and an assumption~$A$ (given by a \tba~$\auta$):
Given $\obs = (\phi_1,I_1,m_1)\cdots (\phi_n,I_n,m_n)$ (an observation) and $t \ge \sup I_n$, do the following:
\begin{enumerate}

    \item Compute~$\terminal{\auta}{\conswords{\auta}{\obs}, t}$. If it has an empty intersection with $\nonempty{\auta}$, then return~$\outofmodel$. This checks whether there is some some finite word that is consistent with the observation and can be extended to satisfy the assumption. If this is not the case, then the assumption was wrong.
    
    \item Compute~$\terminal{\aut_{\neg \varphi} \otimes \auta}{\conswords{\auta}{\obs}, t}$. If it has an empty intersection with $\nonempty{\aut_{\neg \varphi} \otimes \auta}$, then return~$\top$: 
    If there is a finite word consistent with the observation that can be extended to satisfy the assumption, but no such extension satisfies the complement of the property, then every such extension must satisfy the property. Hence, we can return~$\top$.

    \item Compute~$\terminal{\aut_\varphi \otimes \auta}{\conswords{\auta}{\obs}, t}$. If it has an empty intersection with $\nonempty{\aut_\varphi \otimes \auta}$, then return~$\bot$:    
    If there is a finite word consistent with the observation that can be extended to satisfy the assumption, but no such extension satisfies the property, then every such extension must satisfy the complement of the property. Hence, we can return~$\bot$.
    
    \item Return~$\unknown$. Otherwise, there is both a finite word that is consistent with the observation that can be extended to satisfy the property and a finite word that is consistent with the observation that can be extended to satisfy the complement of the property. Consequently, we return~$\unknown$.
\end{enumerate}

\begin{theorem}
    The algorithm described above can be implemented using zones and computes~$\assumpevalfuncsymbol(\varphi,A)$.
\end{theorem}

\begin{proof}
    This follows immediately from Lemma~\ref{lemma_moniviareachsets}, Lemma~\ref{lemma:reachset}, and the fact that non-emptiness of intersections of sets of symbolic states can be solved algorithmically~\cite{dp}.
\end{proof}

Note that our algorithm is online in the following sense: The set of non\-empty states only needs to be computed once for each of the three automata and the symbolic states capturing \[\terminal{\auta}{\conswords{\auta}{(\phi_1,I_1,m_1)\cdots (\phi_n,I_n,m_n)(\phi_{n+1},I_{n+1},m_{n+1})}, t'}\] can be computed from the symbolic states capturing
\[\terminal{\auta}{\conswords{\auta}{(\phi_1,I_1,m_1)\cdots (\phi_n,I_n,m_n)}, t},\]
as evident from the proof of Lemma~\ref{lemma:reachset}.
The same is true for the reach-sets in the other two automata~$\aut_{\neg \varphi} \otimes \auta$ and $\aut_{ \varphi} \otimes \auta$.

\begin{remark}
Our algorithm requires \tba both for the property \emph{and} its complement, but \tba are in general not closed under complementation~\cite{alur1994tba}.
For the important case of \mitl properties, such automata always exist, as \mitl is closed under negation and can be translated into equivalent \tba (Theorem~\ref{thm:mtltba}).
In general, one needs both automata for the property and the complement of the property for monitoring of real-time properties~\cite{concur25}.
\end{remark}

Note that online monitoring as described above is not a decision problem per se. 
Instead, our algorithm continually updates reach-sets and intersects them with precomputed sets of states with non-empty language to determine verdicts.
Furthermore, the number of zones making up the reach-sets can grow exponentially in the number of observations, even if all events are fully observable. 
Also, the number of zones constituting the reach-sets depends not only on the automata for the property being monitored, but also on the timestamps appearing in the observations. 

Hence, we need an experimental evaluation to determine the true efficiency of our algorithm. To this end, we report in the next section on a prototype implementation, with a particular focus on the time it takes to compute the verdict after a new observation and the growth rate of the zones constituting the reach-sets. 

%%%%%%%%%%%%%%%%%%%%%%%%%%%%%%%%%%%%%%%%%%%%%%%%%%%%%%%%%%%%%%%%%%%%%%%%%%%
%%%%%%%%%%%%%%%%%%%%%%%%%%%%%%%%%%%%%%%%%%%%%%%%%%%%%%%%%%%%%%%%%%%%%%%%%%%
%%%%%%%%%%%%%%%%%%%%%%%%%%%%%%%%%%%%%%%%%%%%%%%%%%%%%%%%%%%%%%%%%%%%%%%%%%%
\section{Evaluation}
\label{sec:eval}

We implemented our assumption-based online monitoring algorithm described in Section~\ref{sec:impl} by extending the \textsc{UPPAAL} tool component~\textsc{MoniTAal}\footnote{\url{https://github.com/DEIS-Tools/MoniTAal}}, thereby demonstrating how the use of assumptions and unobservable events can enhance monitoring capabilities. 
The tool can be utilized as a \texttt{C++} library to monitor properties described by \tba.
It supports intersection, so that given a \tba for the property, one for the negation of the property, as well as one for the assumption, one can create monitors for the assumption, and for the intersection of the assumption and the two property automata, respectively. The tool can then be used to update the reach-set of each automaton (triggered by (possibly inexact) observations) and detect if there is a non-empty intersection with the non-empty language states.
In the following, we report on three proof-of-concept cases, that each demonstrate unique features of monitoring with assumptions.
\begin{itemize}
    \item Task Sequence: This example demonstrates how utilizing assumptions can result in earlier verdicts. Furthermore, it shows the effect of unobservable events on the response time, the time between receiving an event and outputting a verdict. 
    \item Conveyor Belt: This example shows that it is possible to monitor properties of unobservable propositions using assumptions. 
    \item Jobshop Scheduling: This example demonstrates how the size of the automata (and the amount of nondeterministic choices in them) affects the monitoring response time and the number of zones constituting the reach-sets.
\end{itemize}

\subsection{Task Sequence}
We first experiment with a system under monitoring that produces a finite sequence of events $a_1,\ldots,a_k$.
Each $a_i$, with $1\leq i<k$, is followed by $a_{i+1}$ with a time within the interval $[l_i,u_i]$.
The assumption is formalized by the \tba shown
in Figure~\ref{fig:bounded-response-assumption}.
The domain is parameterized on $k$, and the $l_i$ and $u_i$. Further, depending on the experiment, not all the $a_i$ will be observable.
We consider the bounded response property~$G(a_1\rightarrow F_{[0,B]}a_k)$. 
Suppose that we observe a timed word~$(a_1,t_1)\cdots(a_k,t_k)$. If for some $j$ in the range~$1< j\leq k$, we have
$t_j+\sum_{j\leq i< k}u_i\leq B + t_1$, then the verdict at time
$t_j$ is $\top$. On the other hand, if $t_j+\sum_{j\leq
  i< k}l_i> B + t_1$, then the verdict at time $t_j$ is
$\bot$. As a corner case, if $\sum_{1\leq i< k}u_i\leq B$
or $\sum_{1\leq i< k}l_i> B$, the verdict is respectively $\top$
and $\bot$ at time~$0$, since all words of the assumption
respectively satisfy and violate the property.
We run several experiments to show the effect of the assumption and study the scalability under a sequence of unobservable events.

\begin{figure}[t]
    \centering

    \begin{tikzpicture} [ initial text={},thick,>=stealth]
      \node[state, initial] (q0)   at (0,0) {$q_0$};
      \node[state] (q1) at (2,0) {$q_1$};
      \node[state] (q2) at (4.5,0)  {$q_2$};
      \node (dots) at (5.75,0) {$\cdots$};
      \node[state,inner sep =0] (q3) at (7,0)  {$q_{k-1}$};
      \node[state,accepting] (q4)  at (10.7,0) {$q_k$};
      
\draw[->, thick] (q2) edge (dots);
     \draw[->, thick] (dots) edge (q3);

     \draw[->, thick] (q0) edge node[above] {$a_1$} node[below]{$x:=0$} (q1);
     \draw[->, thick] (q1) edge node[above] {$a_2$} node[below,align=center]{$x\in [l_1,u_1]$\\$ x:=0$} (q2);
    \draw[->, thick] (q3) edge node[above] {$a_k$} node[below,align=center]{$x\in [l_{k-1},u_{k-1}]$\\$ x:=0$} (q4);
     \draw[->,loop above] (q4) edge node[above] {$\$ $ } ();     

\end{tikzpicture}

    \caption{A \tba representing the assumption for the bounded
      response example.}
    \label{fig:bounded-response-assumption}
\end{figure}

First we show how unobservable events can affect the response time, the time between receiving an event and outputting a verdict. We pick~$k=100$ and $l_i = 50$ and $u_i = 100$ for all $i$. The events $\{a_i \mid i \in \{21,\ldots, 40\} \cup \{ 61,\ldots, 80\} \}$ are unobservable within the interval $[0, 10000]$. 
In Figure~\ref{fig:varying-resp-time} we see that for each consecutive unobservable event, the response time grows linearly. This is due to the reach-set growing. Nevertheless, the reach-set shrinks when an observable event is received. The minimum response time is 5 microseconds ($\mu$s), the maximum is 116 $\mu$s and the average is 25 $\mu$s.
For reference, if we monitor 5000 consecutive unobservable events, the maximum response time is 32 milliseconds.

\begin{figure}[t]
  \centering
      \begin{tikzpicture}
      \begin{axis}[
        height=6cm,
        xtick={0,20,40,60,80,100},
        xlabel={Observation index},
        ylabel={Response time (ns)}
      ]
      \addplot table [x=a, y=b, col sep=comma] {resp_time_steps1.csv};
      \end{axis}
    \end{tikzpicture}
\caption{Response time of the monitoring implementation, when monitoring the task sequence example with $\{a_i \mid i \in \{21,\ldots, 40\} \cup \{ 61,\ldots, 80\} \}$ being unobservable.}
 \label{fig:varying-resp-time}
\end{figure}    

\begin{figure}
\centering
  \captionof{table}{Verdict distribution for monitoring the task sequence example a thousand times with $k = 10$, $B = 675$, $l_i = 50$ and and $u_i = 100$ for all $i$, with and without assumption. 
For the definitive verdicts~$\top$ and $\bot$, the entry in the row labeled $a_i$ specifies the number of new verdicts given after the $i$-th observation; 
for the inconclusive verdict~$\unknown$, the entry specifies how many of the thousand example observations still yield that verdict.
  %Each row specifies the number of times each verdict is given after an observation.
  }
  \begin{tabular}{c >{\centering}m{20pt}>{\centering}m{20pt}>{\centering}m{20pt} >{\centering}m{20pt}>{\centering}b{20pt}>{\centering\arraybackslash}m{20pt}}
       Observation & \multicolumn{6}{c}{Verdicts}\\ \cmidrule(r){2-7}
        & \multicolumn{3}{c}{No Assumption} & \multicolumn{3}{c}{With Assumption}\\
         \cmidrule(r){2-4} \cmidrule(r){5-7}
         & $\top$ & $\bot$ & $\unknown$ & $\top$ & $\bot$ & $\unknown$ \\
         
        \midrule
            $a_{5}$ & 0 & 0 & 1000 & 0 & 0 & 1000\\
      \rowcolor{lightgray!40}  $a_{6}$ & 0 & 0 & 1000 & 0 & 1 & 999\\
        $a_{7}$ & 0 & 0 & 1000 & 17 & 15 & 967\\
      \rowcolor{lightgray!40}  $a_{8}$ & 0 & 0 & 1000 & 81 & 90 & 796\\
        $a_{9}$ & 0 & 33 & 967 & 165 & 153 & 478\\
      \rowcolor{lightgray!40}  $a_{10}$ & 0 & 457 & 510 & 246 & 232 & 0\\
    \end{tabular}
  
  \label{tab:bounded}
\end{figure}

To show the effect of the assumption, we monitor the bounded response property $G(a_1\rightarrow F_{[0,675]}a_{10})$ a thousand times, with and without the assumption where $l_i = 50$ and $u_i = 100$ for all $i$. The observed words are random, but within the assumption.
The results in Table~\ref{tab:bounded} show that verdicts are computed earlier with the assumption than without. Without the assumption the earliest verdicts were in 33 cases $\bot$ after observing $a_9$, while with the assumption we saw a $\top$ or $\bot$ verdict in 522 cases before observing $a_{10}$.

Thus, when monitoring a live system in an online setting (compared to evaluating a log history), a verdict can be reached earlier, because of the restrictions the assumption inhibits. 
Furthermore, we demonstrate in this case how unobservable events can affect the size of the reach-set, as the number of words that are consistent with an observation can increase with the number of consecutive unobservable events. This in turn affects the response time.

\subsection{Conveyor Belt}

This example represents a conveyor belt that moves an item through
different stations, where the item is processed according to some
task. The task in the nominal case takes between $8$ and $10$ times
units. However, if the process is faulty, it finishes earlier and
takes between $7$ and $9$ time units: it may sometimes complete
correctly on time, but it may in other cases stop too early. The fault
can happen at any time and is permanent. Our 
assumption automaton is shown in Figure~\ref{fig:belt-assumption}. Our
monitoring property is simply $G\neg fault$.

\begin{figure}[bt]
    \centering
    \begin{tikzpicture} [node distance = 4.0cm, initial text={},thick, >=stealth]
      \node (q0)  [state, initial, accepting] {$q^n_0$};
      \node (q1)  [state, right = of q0] {$q^n_1$};
      \node (q2)  [state, right = of q1] {$q^n_2$};
      \node (qf0)  [state, accepting, below = 1cm of q0] {$q^f_0$};
      \node (qf1)  [state, right = of qf0] {$q^f_1$};
      \node (qf2)  [state, right = of qf1] {$q^f_2$};
    
     \draw[->, thick] (q0) edge node[above] {$start$} node[below]{$x=1 / x:=0$} (q1);
     \draw[->, thick] (q1) edge node[above] {$stop$} node[below]{$x\in [8,10]/ x:=0$} (q2);
     \draw[->, thick] (q2) edge [bend right=22, above] node[above] {$move$} node[below]{$x=1 / x:=0$} (q0);
     
     \draw[->, thick] (qf0) edge node[above] {$start$} node[below]{$x=1 / x:=0$} (qf1);
     \draw[->, thick] (qf1) edge node[above] {$stop$} node[below]{$x\in [7,9] / x:=0$} (qf2);
     \draw[->, thick] (qf2) edge [bend left=22, below] node[above] {$move$} node[below]{$x=1 / x:=0$} (qf0);
     \draw[->, thick] (q0) edge node[left] {$fault$} (qf0);
     \draw[->, thick] (q1) edge node[left] {$fault$} (qf1);
     \draw[->, thick] (q2) edge node[left] {$fault$} (qf2);

\end{tikzpicture}
    \caption{A \tba representing the assumption for the conveyor
      belt example.} 
    \label{fig:belt-assumption}
\end{figure}

Consider that we observe the events $start$, $stop$, and $move$ with
precise information on the time. If the $stop$ signal happens less
than $8$ time units after $start$, we detect a violation of the
property. If instead $stop$ happens between $8$ and $9$ time units, we
cannot say if there was a fault or not.

Suppose now that we have uncertainty on the time of the
observations like in the following observation sequence: {\small
\begin{align*}(start,[1,1],=1)(fault,[1,11],\ge 0)(stop,[8,10],=1)(fault,[8,11],\ge 0)(move,[9,11],=1)\\(fault,[9,12],\ge 0)(start,[10,12],=1)(fault,[11,22],\ge 0)(stop,[16,18],=1)\end{align*}}
The first $stop$ happens at a time between $8$ and $10$, thus between
$7$ and $9$ time units after the first $start$. This is compatible
with both a nominal (with $stop$ occurring between times $9$ and $10$)
and a faulty execution (with $stop$ occurring between times $8$ and
$9$): after the first stop, we do not know if there was a
fault.

The second $start$ happens in the time interval $[10,12]$ and has the
same uncertainty: it is consistent with the nominal behavior if
$start$ actually occurred in $[11,12]$  and with a faulty behavior if
$start$ occurred in $[10,11]$. 
The second stop happens in the time interval $[16,18]$. Thus, the
difference with the previous start is between $4$ and $8$ time
units. This seems compatible with a nominal delay ($[8,10]$). However,
from the reasoning done above, if there were no fault the second start
would have occurred in the interval $[11,12]$ and the second $stop$
would have occurred in the interval $[19,22]$, which is not compatible
with the observation. Thus we can conclude there was a fault.

% \begin{table}[bt]
%     \centering
%     \caption{Distribution of verdicts when monitoring 1000 random words of the conveyor belt assumption. Each column shows the number of times a conclusive verdict is given after observing the pattern a number of times. The longest run had 24 repetitions of the pattern before a verdict is given.}
%     \label{tab:conveyor}\smallskip
%     \setlength{\tabcolsep}{7pt}
%     \begin{tabular}{lrrrrrrrrrrrrrrrrrrrrrrrr}
%     Repetitions & ${1}$ & ${2}$ & ${3}$ & ${4}$ & ${5}$ & ${6}$ & ${7}$ & ${8}$ & ${9}$ & ${10}$ & ${11}$ \\
% \#Verdicts & 251 & 185 & 125 & 121 & 90 & 47 & 48 & 31 & 24 & 20 & 13 \\
% \midrule
% Repetitions & ${12}$ & ${13}$ & ${14}$ & ${15}$ & ${16}$ & ${17}$ & ${18}$ & ${19}$ & ${20}$ & ${24}$&\\
% \#Verdicts  & 13 & 9 & 4 & 5 & 4 & 2 & 2 & 2 & 3 & 1&\\
%     \end{tabular}
%     \vspace{-4mm}
% \end{table}

\begin{figure}
    \centering
    \begin{tikzpicture}
    \begin{axis}[
        height=0.5\linewidth,
        width=\linewidth,
    	x tick label style={
    		/pgf/number format/1000 sep=10},
        xtick={2, 10, 20, 30, 40, 50, 60, 70, 77},
    	xlabel=Number of observations before a fault occurred,
        ylabel=Number of runs,
    	enlargelimits=0.03,
    	ybar=0pt,
        bar width=1,
    ]
    \addplot coordinates {(2, 93) (3, 73) (4, 56) (5, 99) (6, 55) (7, 36) (8, 78) (9, 40) (10, 35) (11, 51) (12, 36) (13, 20) (14, 41) (15, 26) (16, 13) (17, 26) (18, 21) (19, 14) (20, 20) (21, 15) (22, 5) (23, 24) (24, 9) (25, 6) (26, 13) (27, 7) (28, 6) (29, 19) (30, 7) (31, 2) (32, 10) (33, 5) (34, 2) (35, 4) (36, 3) (37, 2) (38, 1) (39, 2) (40, 1) (41, 2) (42, 1) (43, 2) (44, 0) (45, 0) (46, 0) (47, 2) (48, 1) (49, 3) (50, 2) (51, 1) (52, 3) (53, 1) (54, 0) (55, 0) (56, 1) (57, 1) (58, 0) (59, 0) (60, 0) (61, 1) (62, 0) (63, 0) (64, 0) (65, 1) (66, 0) (67, 0) (68, 0) (69, 0) (70, 0) (71, 1) (72, 0) (73, 0) (74, 0) (75, 0) (76, 0) (77, 1)};
    \end{axis}
    \end{tikzpicture}
    \caption{Histogram of 1000 runs showing how many \textit{start}, \textit{stop}, and \textit{move} events are observed before a fault is detected while monitoring a faulty conveyor belt.}
    \label{fig:conveyor-histogram}
\end{figure}

% \begin{figure}
%     \centering
%     \begin{tikzpicture}
%     \begin{axis}[
%         height=0.5\linewidth,
%         width=\linewidth,
%     	x tick label style={
%     		/pgf/number format/1000 sep=10},
%         xtick={2, 10, 20, 30, 40, 50, 60, 70, 74},
%     	ylabel=Observations,
%     	enlargelimits=0.03,
%     	ybar=0pt,
%         bar width=1,
%     ]
%     \addplot coordinates {(2, 324) (3, 0) (4, 0) (5, 225) (6, 0) (7, 0) (8, 143) (9, 0) (10, 0) (11, 99) (12, 0) (13, 0) (14, 71) (15, 0) (16, 0) (17, 46) (18, 0) (19, 0) (20, 31) (21, 0) (22, 0) (23, 22) (24, 0) (25, 0) (26, 12) (27, 0) (28, 0) (29, 11) (30, 0) (31, 0) (32, 5) (33, 0) (34, 0) (35, 2) (36, 0) (37, 0) (38, 5) (39, 0) (40, 0) (41, 2) (42, 0) (43, 0) (44, 1) (45, 0) (46, 0) (47, 0) (48, 0) (49, 0) (50, 0) (51, 0) (52, 0) (53, 0) (54, 0) (55, 0) (56, 0) (57, 0) (58, 0) (59, 0) (60, 0) (61, 0) (62, 0) (63, 0) (64, 0) (65, 0) (66, 0) (67, 0) (68, 0) (69, 0) (70, 0) (71, 0) (72, 0) (73, 0) (74, 1)};
%     \end{axis}
%     \end{tikzpicture}
%     \caption{Histogram over the number of \textit{start}, \textit{stop}, and \textit{move} events that was observed while monitoring a faulty conveyor belt, before a fault was finally detected.}
%     \label{fig:conveyor-histogram-no-uncertainty}
% \end{figure}

Using \textsc{MoniTAal}, we monitored the property $G\neg fault$ with the assumption from Figure~\ref{fig:belt-assumption} by simulating a faulty conveyor belt. The observed trace is generated according to the simulated faulty behavior of the assumption, and the monitor observes only \textit{start}, \textit{stop}, and \textit{move} signals, with a time uncertainty of 2. For example, if a \textit{stop} occurs at time \textit{t} in the simulation, then the monitor would observe $(\textit{stop}, [l, u], =1)$ for some randomly selected $l$ and $u$ such that $u - l = 2$ and $l \le t \le u$.
% with an unbounded repeating pattern $\rho_1 \cdot_{\tau(\rho_1)} \rho_2 \cdot_{\tau(\rho_1)+\tau(\rho_2)} \cdots$ with each $\rho_i$ having the form
% \[ (fault, [0, 1],\ge 0), (start, 1,=1), (fault, [1, t_i+2],\ge 0), (stop, [t_i, t_i+2],=1), (move, t_i+2,=1)\] for some uniformly chosen $t_i \in \{7,8,9,10\}$.
As our observations are generated from the assumption, it is never violated.
Furthermore, for the property~$G \neg \textit{fault}$, a monitor can never give the verdict~$\top$. Thus, the only conclusive verdict we report is $\bot$ i.e., that the property does not hold. 
% The pattern essentially randomly selects whether $stop$ is observed after 7, 8, 9 or 10 time units after $start$. Since 7 is only possible after a fault, there is a 1 in 4 chance, per repetition, of violating the property.

The histogram in Figure~\ref{fig:conveyor-histogram} shows the distribution of the number of \textit{start}, \textit{stop}, and \textit{move} signals that were observed before a fault was detected for 1000 runs. We see that the fault is often detected relatively early, even though there is a time uncertainty in every observation.
The trace always begins with a \textit{start} signal. Since the constraints on \textit{start} are the same in the faulty behavior, a fault is never detected after just one observation, but a fault was detected after the second observation in 93 cases.
The longest trace had 77 observations before a fault was detected. 
% in 251 out of 1000 cases a definitive verdict is given after observing the pattern once, and that the longest is 24 repetitions.

With this example, we see how an assumption makes it possible to monitor properties over unobservable events (recall what we do not observe the signal~$\textit{fault}$ referred to in the property). Without an assumption, reasoning about unobservable behavior would not be possible for such a property.

\subsection{Jobshop Scheduling}
\label{sec:jobshop}
In this example, we test the capability of the assumption-based online
monitoring of scaling with the number of non-deterministic choices that
the monitored system may take.

%% The assumption is intuitively the product composition of $n+1$
%% automata that represent some processes trying to complete a task
%% accessing two shared resources, $A$ and $B$, concurrently.  Thus, the
%% processes must lock a resource (either $A$ or $B$) within $n$ time
%% units, and take some time to complete the task. Process $0$ takes $n$
%% time units to complete, while the others only $1$ time unit. We
%% monitor the property $F_{[0,n]} done$, where \textit{done} is the
%% global state in which all tasks are completed.

%% The assumption is intuitively the
%% product composition of $n+3$ automata depicted in
%% Figure~\ref{fig:jobshop}. Two of them (at the bottom of the figure)
%% represent two resources, $A$ and $B$, that may be locked by the other
%% $n+1$ automata concurrently. The other automata represent some
%% processes that, within $n$ time units, must lock a resource and take
%% some time to complete the task. Process $0$ takes $n$ time units to
%% complete, while the others only $1$ time unit. We monitor the property
%% $F_{[0,n]} done$, where \textit{done} is the global state in which all
%% tasks are completed.

The assumption is intuitively the product composition of $n+1$
automata depicted in Figure~\ref{fig:jobshop-network}. These automata
represent some processes that, within $n$ time units, must lock a
resource (either $A$ or $B$) and take some time to complete the
task. Thus, a process can use a resource only if no other process is
already using it (this is intuitively represented by the guard $\neg
A$ or $\neg B$). Process $0$ takes $n$ time units to complete, while
the others only $1$ time unit. We monitor the property $F_{[0,n]}
done$, where \textit{done} is the global state in which all tasks are
completed.

\begin{figure}[h]
    \centering
    \begin{tikzpicture} [node distance = 4cm, initial text={},thick, >=stealth]
      \node (q00)  [state, initial] {$p^I_0$};
      \node (q01)  [state, right = of q00, yshift=1.5cm] {$p_0^A$};
      \node (q02)  [state, right = of q00, yshift=-1.5cm] {$p_0^B$};
      \node (q03)  [state, right = of q02, yshift=1.5cm] {$p_0^D$};

      %% \node (q10)  [state, initial, below = of q00] {$p^I_1$};
      %% \node (q11)  [state, right = of q10, yshift=1cm] {$p_0^A$};
      %% \node (q12)  [state, right = of q10, yshift=-1cm] {$p_0^B$};
      %% \node (q13)  [state, right = of q12, yshift=1cm] {$p_0^D$};

      %% \node (qdot) [draw=none, below = 1.2cm of q10] {$\ldots$};
      
      \node (qn0)  [state, initial, below = 3.4cm of q00] {$p^I_i$};
      \node (qn1)  [state, right = of qn0, yshift=1.5cm] {$p_i^A$};
      \node (qn2)  [state, right = of qn0, yshift=-1.5cm] {$p_i^B$};
      \node (qn3)  [state, right = of qn2, yshift=1.5cm] {$p_i^D$};

      \draw[->, thick,bend left=15] (q00) edge node[above] {$\tau$}
      node[below,align=center,yshift=-.2cm]{$\neg A\wedge x_0\leq n $\\ $x_0:=0$} (q01);
      \draw[->, thick,bend right=15] (q00) edge node[above] {$\tau_{B0}$}
      node[below,align=center,yshift=-.2cm]{$\neg B\wedge x_0\leq n$ \\ $x_0:=0$} (q02);
      \draw[->, thick,bend left=15] (q01) edge node[above] {$d_0$}
      node[below,align=center,yshift=-.2cm]{$x_0\geq n$} (q03);
      \draw[->, thick,bend right=15] (q02) edge node[above] {$d_0$}
      node[below,align=center,yshift=-.2cm]{$x_0\geq n$} (q03);

      %% \draw[->, thick] (q10) edge node[above] {$\tau_{A1}$}
      %% node[below]{$x_1\leq n / x_1:=0$} (q11);
      %% \draw[->, thick] (q10) edge node[above] {$\tau_{B1}$}
      %% node[below]{$x_1\leq n / x_1:=0$} (q12);
      %% \draw[->, thick] (q11) edge node[above] {$d_{A1}$}
      %% node[below]{$x_1\geq 1$} (q13);
      %% \draw[->, thick] (q12) edge node[above] {$d_{B1}$}
      %% node[below]{$x_1\geq 1$} (q13);

      \draw[->, thick,bend left=15] (qn0) edge node[above] {$\tau$}
      node[below,align=center,yshift=-.2cm]{$\neg A\wedge x_i\leq n$ \\ $x_i:=0$} (qn1);
      \draw[->, thick,bend right=15] (qn0) edge node[above] {$\tau_i$}
      node[below,align=center,yshift=-.2cm]{$\neg B\wedge x_i\leq n$ \\ $x_i:=0$} (qn2);
      \draw[->, thick,bend left=15] (qn1) edge node[above] {$d_i$}
      node[below,align=center,yshift=-.2cm]{$x_i\geq 1$} (qn3);
      \draw[->, thick,bend right=15] (qn2) edge node[above] {$d_i$}
      node[below,align=center,yshift=-.2cm]{$x_i\geq 1$} (qn3);

      %% \node (r00)  [state, initial, below = 1.8 cm of qn0] {$free_A$};
      %% \node (r01)  [state, right = of r00] {$locked_A$};

      %% \draw[->>>, thick] (r00) edge [bend left,above] node[above] {$\tau_{Ai}$}  (r01);
      %% \draw[->>>, thick] (r01) edge [bend left,below] node[below] {$d_{Ai}$}  (r00);

      %% \node (r10)  [state, initial, right = 1cm of r01] {$free_B$};
      %% \node (r11)  [state, right = of r10] {$locked_B$};

      %% \draw[->>>, thick] (r10) edge [bend left,above] node[above] {$\tau_{Bi}$}  (r11);
      %% \draw[->>>, thick] (r11) edge [bend left,below] node[below] {$d_{Bi}$}  (r10);

\end{tikzpicture}
    \caption{A network of automata representing the assumption
      automaton of the jobshop scheduling problem. The upper automaton
    represents process 0, while the lower automaton represents process
    $i$, for $1\leq i\leq n$.}
    \label{fig:jobshop-network}
\end{figure}

In practice, the assumption is defined as a single timed automaton, which
intuitively corresponds to the product of the automata depicted in
Figure~\ref{fig:jobshop-network}. This is shown in
Figure~\ref{fig:jobshop2} for the simple case of $n=1$.

More precisely, the assumption automaton has $4^{(n+1)}$ states
represented by a tuple $\langle s_0,s_1,\ldots,s_n\rangle$, where
every $s_i\in \{I,A,B,D\}$, where $I$ stands for idle, $A$ for
processing the task with resource $A$, $B$ for processing with
resource $B$, and $D$ for done (corresponding to the locations~$p_i^I$, $p_i^A$, $p_i^B$, and $p_i^D$ in Figure~\ref{fig:jobshop-network}). Thus, \textit{done} is the state in
which $s_i=D$, for all $i$, while the initial state is the state in
which $s_i=I$ for all $i$. The automaton has $n+1$ clocks $x_i$. The
transitions are defined as follows:
\begin{itemize}
\item
  $\langle s_0,s_1,\ldots,I,\ldots,s_n\rangle
  \labeledtransition{\tau,x_i\leq n,\{x_i\}} \langle
  s_0,s_1,\ldots,A,\ldots,s_n\rangle$ switching $s_i$ from $I$ to $A$
  if $s_j\neq A$ for all $j$, with $0\leq j\leq n$ and $j\neq i$;
\item
  $\langle s_0,s_1,\ldots,I,\ldots,s_n\rangle
  \labeledtransition{\tau,x_i\leq n,\{x_i\}} \langle
  s_0,s_1,\ldots,B,\ldots,s_n\rangle$ switching $s_i$ from $I$ to $B$
  if $s_j\neq B$ for all $j$, with $0\leq j\leq n$ and $j\neq i$;
\item
  $\langle s_0,s_1,\ldots,s_n\rangle
  \labeledtransition{d_0,x_i\geq n} \langle
  D,s_1,\ldots,s_n\rangle$ switching $s_0$ to $D$
  if $s_0\in\{A,B\}$;
\item
  $\langle s_0,s_1,\ldots,s_n\rangle
  \labeledtransition{d_i,x_i\geq 1} \langle
  s_0,s_1,\ldots,D,\ldots,s_n\rangle$ switching $s_i$ to $D$
  if $s_i\in\{A,B\}$ and $1\leq i\leq n$.
\end{itemize}

Let us suppose that $\tau$ is not observable and $d_i$ is observable
for all $i$, $0\leq i\leq n$. Then $F_{[0,n]} done$ is monitorable as
the monitor's output must be $\top$ after observing all $d_i$, one
after the other, in any order, if the last process completes within
$n$ time units. The output must be $\bot$ if after $n$ time units some
$d_i$ has not been observed.

Given the above assumption, we can detect a violation of the property
in advance. In fact, note that the property can be satisfied only if
process $p_0$ starts at time $0$ choosing a resource and the other
processes one after the other, starting at time $0$, complete the task
choosing the other resource. For example, with $n=2$, a satisfying run
would be
\begin{align*}
&(\langle I,I,I\rangle, x_0=0,x_1=0,x_2=0) \labeledtransition{(\tau,0)} 
(\langle I,A,I\rangle, x_0=0,x_1=0,x_2=0) \labeledtransition{(\tau,0)}\\ 
&(\langle B,A,I\rangle, x_0=0,x_1=0,x_2=0) \labeledtransition{(d_1,1)} 
(\langle B,D,I\rangle, x_0=1,x_1=1,x_2=1) \labeledtransition{(\tau,1)} \\
&(\langle B,D,A\rangle, x_0=1,x_1=1,x_2=0) \labeledtransition{(d_2,2)} 
(\langle B,D,D\rangle, x_0=2,x_1=2,x_2=1)\labeledtransition{(d_0,2)} \\
&(\langle D,D,D\rangle, x_0=2,x_1=2,x_2=1)\vphantom{\labeledtransition{(d_0,2)}}.
\end{align*}

Instead, if we observe $d_1$ and $d_2$ to occur before time $2$, from the
assumption, we can deduce that $p_1$ and $p_2$ used different
resources. Thus, $p_0$ cannot complete within time $n$ and the
property will be violated.

\begin{figure}[h]
    \centering\small
    \begin{tikzpicture} [node distance = 1cm, initial text={},thick, >=stealth]

      \node (qII)  [state, initial above] {$\langle I,I\rangle$};

      \node (qAI)  [state, below=of qII, xshift=-3.2cm] {$\langle A,I\rangle$};
      \node (qBI)  [state, right=of qAI] {$\langle B,I\rangle$};
      \node (qIB)  [state, right=of qBI] {$\langle I,B\rangle$};
      \node (qIA)  [state, right=of qIB] {$\langle I,A\rangle$};

      \node (qDI)  [state, below=of qAI] {$\langle D,I\rangle$};
      \node (qAB)  [state, right=of qDI] {$\langle A,B\rangle$};
      \node (qBA)  [state, right=of qAB] {$\langle B,A\rangle$};
      \node (qID)  [state, right=of qBA] {$\langle I,D\rangle$};

      \node (qDB)  [state, below=of qDI] {$\langle D,B\rangle$};
      \node (qDA)  [state, right=of qDB] {$\langle D,A\rangle$};
      \node (qAD)  [state, right=of qDA] {$\langle A,D\rangle$};
      \node (qBD)  [state, right=of qAD] {$\langle B,D\rangle$};

      \node (qDD)  [state, below=of qDA, xshift=1cm] {$\langle D,D\rangle$};
      
      \draw[->, thick] (qII) edge node[left,near start] {$\tau$} (qAI) ;
      \draw[->, thick] (qII) edge node[left,near start] {$\tau$} (qBI);
      \draw[->, thick] (qII) edge node[right,near start] {$\tau$} (qIA);
      \draw[->, thick] (qII) edge node[right,near start] {$\tau$} (qIB);
      
      \draw[->, thick] (qAI) edge node[left,near start] {$d_0$} (qDI);
      \draw[->, thick] (qAI) edge node[left,near start] {$\tau$} (qAB);
      \draw[->, thick] (qBI) edge node[right,near start] {$d_0$} (qDI);
      \draw[->, thick] (qBI) edge node[left,near start] {$\tau$} (qBA);
      \draw[->, thick] (qIB) edge node[right,near start] {$\tau$} (qAB);
      \draw[->, thick] (qIB) edge node[left,near start] {$d_1$} (qID);
      \draw[->, thick] (qIA) edge node[right,near start] {$\tau$} (qBA);
      \draw[->, thick] (qIA) edge node[right,near start] {$d_1$} (qID);

      \draw[->, thick] (qDI) edge node[left,near start] {$\tau$} (qDB);
      \draw[->, thick] (qDI) edge node[left,near start] {$\tau$} (qDA);
      \draw[->, thick] (qAB) edge node[right,near start] {$d_0$} (qDB);
      \draw[->, thick] (qAB) edge node[left,near start] {$d_1$} (qAD);
      \draw[->, thick] (qBA) edge node[right,near start] {$d_0$} (qDA);
      \draw[->, thick] (qBA) edge node[left,near start] {$d_1$} (qBD);
      \draw[->, thick] (qID) edge node[right,near start] {$\tau$} (qAD);
      \draw[->, thick] (qID) edge node[right,near start] {$\tau$} (qBD);

      \draw[->, thick] (qDB) edge node[left,near start] {$d_1$} (qDD);
      \draw[->, thick] (qDA) edge node[left,near start] {$d_1$} (qDD);
      \draw[->, thick] (qAD) edge node[right,near start] {$d_0$} (qDD);
      \draw[->, thick] (qBD) edge node[right,near start] {$d_0$} (qDD);

    \end{tikzpicture}
    \caption{Jobshop assumption for $n=1$. Clocks guards, resets, and unreachable locations are
      omitted to improve readability.}
    \label{fig:jobshop2}
\end{figure}

For the experiments, we will monitor a satisfying observation of the form 
\[(\tau, [0, n], \ge 0), (d_1, [1,1], = 1), (\tau, [0,n], \ge 0), (d_2, [2,2], =1), ..., (d_0, [n,n], = 1),\] where $n+1$ is the number of jobs. Table~\ref{tab:jobshop-bench} shows the results of monitoring such a satisfying observation for 2 to 9 jobs. The size of the assumption is the number of locations, the maximum response time is the longest time it took to compute a verdict after receiving an observation, and the maximum number of symbolic states refers to the size of the representation of the combined reach-sets for the property and assumption.
We see that the increasing size of the assumption slows the response time significantly around 9 concurrent jobs, where the response time goes above a second.

\begin{table}[]
    \centering
        \caption{Experimental results for monitoring a satisfying word of the jobshop scheduling example with 2 - 9 jobs. Response time is the maximal time it takes to compute a verdict after an observation. The number of symbolic states refers to the combined size of the reach-set of the property and assumption automata.}
    \label{tab:jobshop-bench}
    \begin{tabular}{crrr}
        Jobs & Size of Assumption  & Max. resp. Time & Max. \# Symbolic States \\\hline
        2 & 14 & 0.1 ms & 9 \\
       \rowcolor{lightgray!40} 3 & 44 & 0.3 ms & 15 \\
        4 & 128 & 1.1 ms & 21 \\
      \rowcolor{lightgray!40}  5 & 352 & 2.8 ms & 27 \\
        6 & 928 & 9.7 ms & 33 \\
      \rowcolor{lightgray!40}  7 & 2.368 & 38.3 ms & 39 \\
        8 & 5.888 & 198.6 ms & 45 \\
      \rowcolor{lightgray!40}  9 & 14.336 & 1.294.6 ms & 51 \\
        
    \end{tabular}
\end{table}

%%%%%%%%%%%%%%%%%%%%%%%%%%%%%%%%%%%%%%%%%%%%%%%%%%%%%%%%%%%%%%%%%%%%%%%%%%%
%%%%%%%%%%%%%%%%%%%%%%%%%%%%%%%%%%%%%%%%%%%%%%%%%%%%%%%%%%%%%%%%%%%%%%%%%%%
%%%%%%%%%%%%%%%%%%%%%%%%%%%%%%%%%%%%%%%%%%%%%%%%%%%%%%%%%%%%%%%%%%%%%%%%%%%
\section{Related Work}
\label{sec:relatedwork}

Our automata-based monitoring of finite words against specifications over infinite words follows the seminal work of Bauer~et al.~\cite{bauer2006monitoring}, who presented monitoring algorithms for LTL and timed LTL.
Their algorithm for timed LTL is based on clock regions~\cite{alur1994tba}, while we follow the approach of Grosen et al.~\cite{GrosenKLZ22} and use clock zones~\cite{DBLP:conf/ac/BengtssonY03}, whose performance is an order of magnitude faster.
Also, they translated timed LTL into event-clock automata, which are less expressive than the timed Büchi automata (\tba) used both by Grosen et al.~\cite{GrosenKLZ22} and here.
This approach has also been applied to monitoring under delayed observations~\cite{fglz24}.

As our algorithms work with \tba, we also support MITL specifications, as these can be compiled into \tba.
The monitoring problem for MITL has been investigated before.
Baldor~et al.\ showed how to construct a monitor for dense-time MITL formulas by constructing a tree of timed transducers~\cite{baldor2013monitoring}.
Ho et al.\ split unbounded and bounded parts of MITL formulas for monitoring, using traditional LTL monitoring for the unbounded parts and permitting a simpler construction for the (finite-word) bounded parts~\cite{ho2014online}.

There is also a large body of work on monitoring with finite-word semantics.
Roşu et al.\ focussed on discrete-time finite-word
MTL~\cite{rosu2005monitoring}, while Basin et al.\ proposed algorithms
for monitoring real-time finite-word
properties~\cite{basin2012algorithms} and compared different time models. Donz\'{e} et al.~\cite{DonzeFM13}
focussed on monitoring a quantitative semantics for STL, a variant of
MTL with predicates over real-valued signals. Andr\'e et al. consider monitoring finite logs of parameterized timed and hybrid systems \cite{DBLP:journals/tcps/WagaAH22}.
Finally, Ulus et al.\ described monitoring timed regular expressions over finite words using unions of two-dimensional zones~\cite{ulus2014offline,ulus2016online}.

The contribution of this paper is focused on extending the monitoring
of timed properties with assumptions, framing the problem as defined
in~\cite{Tian:2019a,Tian:2019b,CimattiTT21,CimattiTT22} for the discrete-time setting.
Assumptions were first used in \cite{Leucker:2012gl} for extending the
monitoring of LTL with predictive capabilities. 
In \cite{Zhang:2018vv}, the assumption for predictive RV is
computed applying static analysis to the monitored program.
Pinisetty et al. further extend the predictive RV idea to support RV
of timed properties~\cite{Pinisetty:2017jx}, where the \emph{a priori
knowledge} is also expressed as a timed property.
As in~\cite{CimattiTT22}, we
adopt a four-valued semantics for timed properties and we support
partial observability.
Besides the complexity of moving from discrete to dense time
semantics, the ABRV framework is extended with a rich notion of
observations that take into account uncertainty on the time.

The research of partial observability in Discrete-Event Systems is
usually connected with diagnosability \cite{sampath1995diagnosability}
and predictability \cite{Genc:2009gk,Genc:2006jd}.
These notions have been extended to timed systems (see,
e.g., \cite{CassezG13,CassezT13}).
Moreover, they are related to monitorability, an important topic in RV
and other related fields~\cite{Aceto:2019,SistlaZF11,Peled:2019iz},
which has been studied taking into account assumptions
in~\cite{Henzinger:2020ds}. Recently, the monitorability problem for real-time properties has been studied~\cite{concur25,monitorability}.

%%%%%%%%%%%%%%%%%%%%%%%%%%%%%%%%%%%%%%%%%%%%%%%%%%%%%%%%%%%%%%%%%%%%%%%%%%%
%%%%%%%%%%%%%%%%%%%%%%%%%%%%%%%%%%%%%%%%%%%%%%%%%%%%%%%%%%%%%%%%%%%%%%%%%%%
%%%%%%%%%%%%%%%%%%%%%%%%%%%%%%%%%%%%%%%%%%%%%%%%%%%%%%%%%%%%%%%%%%%%%%%%%%%
\section{Conclusion}
\label{sec:conc}

We presented an approach to assumption-based runtime verification for
real-time systems under partial observability. The method builds on
Timed Automata to represent assumptions and uses Metric Interval
Temporal Logic (or Timed Automata) to specify properties. A key element of the approach is
the formalization of observations with both data and time uncertainty,
allowing the monitor to reason about incomplete inputs.
We proved that the monitoring function satisfies the properties of
impartiality and anticipation, and we introduced a partial order on
the verdicts to formalize their refinement. We developed a zone-based
online monitoring algorithm and implemented it on top of
UPPAAL. Experimental results illustrate how assumptions can enable
earlier verdicts and allow monitoring of properties involving
unobservable events.

There are several directions for future development. One is to improve
scalability of the algorithm, especially when the assumption is
specified as a network of automata. Another is to extend the output of
the monitor to include richer diagnostic information. It would also be
useful to investigate an epistemic characterization of the monitor
specification, capturing what can be known from the
observations. Finally, applying the approach in real-world settings
would help assess its practical relevance and identify further
challenges.

\backmatter

\bmhead{Supplementary information}

The tool implementing the algorithms presented here is freely available at \url{https://github.com/DEIS-Tools/MoniTAal}.

\bmhead{Acknowledgments}
T.M.\ Grosen and K.G.\ Larsen have been funded by the Villum Investigator Grant S4OS. T.M.\ Grosen, K.G.\ Larsen, and M.\ Zimmermann
    have been supported by DIREC - Digital Research Centre
    Denmark. A.\ Cimatti and S.\ Tonetta have been supported by the PNRR
    project FAIR - Future AI Research (PE00000013), under the Italian NRRP MUR
    program funded by the NextGenerationEU.

\section*{Declarations}

The authors have no competing interests to declare that are relevant to the content of this article.

\bibliography{biblio}

\end{document}